\documentclass[a4paper,UKenglish,cleveref, autoref, thm-restate]{lipics-v2019}

\usepackage{tikz}
\usetikzlibrary{shapes.geometric}
\usepackage{float}
\usepackage{tcolorbox}

\usepackage{enumerate}
\usepackage{graphicx}
\usepackage{amsmath}
\usepackage{amssymb}
\usepackage{comment}

\def \calf {{\cal F}}

\newtheorem{observation}[theorem]{Observation}

\title{Perfectly Matched Sets in Graphs: Parameterized and Exact Computation} 

\titlerunning{Perfectly Matched Sets in Graphs} 

\author{N.R. Aravind}
{Department of Computer Science \& Engineering, Indian Institute of Technology Hyderabad, India}
{aravind@cse.iith.ac.in}{https://orcid.org/0000-0002-6590-7592}
{}
\author{Roopam Saxena}
{Department of Computer Science \& Engineering, Indian Institute of Technology Hyderabad, India}{cs18resch11004@iith.ac.in}{}{}
\authorrunning{N.R. Aravind and R. Saxena} 
\Copyright{Aravind N.R. and Roopam Saxena} 
\titlerunning{Perfectly Matched Sets in Graphs}


\hideLIPIcs
\nolinenumbers

\category{} 

\relatedversion{} 

\supplement{}

\acknowledgements{}

\EventEditors{}
\EventNoEds{1}
\EventLongTitle{}
\EventShortTitle{} 
\EventAcronym{}
\EventYear{}
\EventDate{}
\EventLocation{}
\EventLogo{}
\SeriesVolume{}
\ArticleNo{}

\begin{document}

\maketitle
 
\begin{abstract}
In an undirected graph $G=(V,E)$, we say that $(A,B)$ is a pair of perfectly matched sets if $A$ and $B$ are disjoint subsets of $V$ and every vertex in $A$ (resp. $B$) has exactly one neighbor in $B$ (resp. $A$). The size of a pair of perfectly matched sets $(A,B)$ is $|A|=|B|$. 
The PERFECTLY MATCHED SETS problem is to decide whether a given graph $G$ has a pair of perfectly matched sets of size $k$.
We show that PMS is $W[1]$-hard when parameterized by solution size $k$ even when restricted to split graphs and bipartite graphs. We observe that PMS is FPT when parameterized by clique-width, and give FPT algorithms with respect to the parameters distance to cluster, distance to co-cluster and treewidth.
Complementing FPT results, we show that PMS does not admit a polynomial kernel when parameterized by vertex cover number unless $\text{NP}\subseteq \text{coNP/poly}$.
We also provide an exact exponential algorithm running in time $O^*(1.966^n)$. 
Finally, considering graphs with structural assumptions, we show that PMS remains NP-hard on planar graphs.

\end{abstract}

\section{Introduction}

Consider the following communication problem: we have an undirected graph each of whose nodes can send or receive messages. We wish to assign some nodes as transmitter or receiver, and then test the fidelity of transmission between transmitter-receiver pairs under the following assumptions: (a) there is no interference, i.e. a receiver doesn't get a message from more than one sender; (b) each sender can send at most one message at a time. What is the maximum number of pairs that can be tested simultaneously?
This question was first studied in \cite{SV82}, where the underlying abstract problem was called TR-matching and shown to be NP-complete on 3-regular graphs.

We first formally define the problem PERFECTLY MATCHED SETS.
\begin{tcolorbox}
{
PERFECTLY MATCHED SETS (PMS): \newline
\textit{Input:} An instance $I$ = $(G,k)$, where $G=(V,E)$ is an undirected graph, and $k\in \mathbb{N}$.\newline
\textit{Output:} YES, if $G$ contains two disjoint sets $A$ and $B$ of size $k$ each such that every vertex in $A$ (resp $B$) has exactly one neighbor in $B$ (resp $A$); NO otherwise.
}
\end{tcolorbox}
The above definition is same as that of the TR-matching problem, as introduced in \cite{SV82}; however we have renamed it because of its relation to two recently well-studied problems: MATCHING CUT and PERFECT MATCHING CUT.

\subsection{Our results}

We revisit this problem in the context of designing parameterized and exact algorithms.
In the context of parameterized complexity, the most natural parameter for PMS is the solution size. In Section \ref{W1-hard}, we start by showing that PMS is $W[1]$-hard when parameterized by solution size, even when restricted to split graph and bipartite graphs. This naturally motivates the study of PMS with respect to other structural parameters to obtain tractability. In Section \ref{clique_width}, we observe that PMS is FPT when parameterized by clique-width using Courcelle's Theorem \cite{CourcelleLinear}. This positive result motivated us to look for efficient FPT algorithms for PMS with other structural parameters. In Sections \ref{sec_cluster}, \ref{sec_co-cluster} and \ref{section-tw}, we obtain FPT algorithms for PMS with parameters distance to cluster, distance to co-cluster and tree-width, these parameters are unrelated to each other and are some of the most widely used structural parameters. On the kernelization side, in Section \ref{section-kernel-lb}, we show that there does not exist a polynomial kernel for PMS when parameterized by vertex cover unless $\text{NP}\subseteq \text{coNP/poly}$. This kernelization lower bound is in contrast to the PMS being FPT when parameterized by vertex cover, which is due to PMS being FPT by distance to cluster graph (a generalization of vertex cover). In Section \ref{sec_exact}, using a result of \cite{LT2020}, we present an exact algorithm for PMS which runs in time $O^*(1.966^n)$. Finally, focusing on restricted graph classes, in Section \ref{section-planar-hardness} we show that PMS remains NP-hard when restricted to planar graphs.

We remark that we are also interested in the optimization version of our problem, i.e. finding subsets $(A,B)$ of maximum size such that $A$ and $B$ are perfectly matched; our exact and FPT (distance to cluster, distance to co-cluster and tree-width) algorithms solve the optimization version.

\subsection{Related work}
Given a graph $G$, a partition $(A,B)$ of $V(G)$ is a {\it matching cut} if every vertex in $A$ (resp $B$) has at most one neighbor in $B$ (resp $A$). The MATCHING CUT  problem is then to decide whether a given graph has a matching cut or not.

The MATCHING CUT problem has been extensively studied. Graham \cite{graham1970primitive} discussed the graphs with matching cut under the name of decomposable graphs. Chv{\'{a}}tal \cite{chvatal1984recognizing} proved that MATCHING CUT is NP-Complete for graphs with maximum degree 4. Bonsma \cite{Bonsma} proved that MATCHING CUT is  NP-complete for planar graphs with maximum degree 4 and with girth 5. Kratsch and Le \cite{kratsch16} provided an exact algorithm with running time $O^*(1.4143^n)\footnote{We use $O^*$ notation which suppresses polynomial factors. }$. Komusiewicz, Kratsch and Le \cite{DBLP:journals/dam/KomusiewiczKL20} provided a deterministic exact algorithm with running time $O^*(1.328^n)$. MATCHING CUT problem is also studied in parameterized realm with respect to various parameters in \cite{kratsch16,DBLP:journals/dam/KomusiewiczKL20,aravind2017structural,Gomes-Sau}. Hardness and polynomial time results have also been obtained for various structural assumptions in \cite{DBLP:conf/isaac/LeL16,DBLP:journals/tcs/LeL19,DBLP:conf/cocoon/HsiehLLP19}. Recently, enumeration version of matching cut is also studied \cite{DBLP:conf/stacs/GolovachKKL21}.

A special case of matching cut where for the partition $(A,B)$, every vertex in $A$ (resp $B$) has exactly one neighbor in $B$ (resp $A$) called perfect matching cut was studied by Heggernes
and Telle \cite{DBLP:journals/njc/HeggernesT98}, where they proved that PERFECT MATCHING CUT problem is NP-complete. Recently, Le and Telle \cite{LT2020} revisited the PERFECT MATCHING CUT problem and showed that it remains NP-complete even when restricted to bipartite graphs with maximum degree $3$ and arbitrarily large girth. 
They  also obtained an exact algorithm running in time $O^*(1.2721^n)$ for PERFECT MATCHING CUT \cite{LT2020}.

Observe that the PERFECT MATCHING CUT problem is more closely related to the PERFECTLY MATCHED SETS problem; indeed the later (with inputs $G,k$) is equivalent to deciding whether the given graph $G$ contains an induced subgraph of size $2k$ that has a perfect matching cut.

 For a graph $G$, a matching $M\subseteq E(G)$ is an induced matching if $(V(M),M)$ is an induced subgraph of $G$.  
 The problem of finding maximum induced matching in a graph is INDUCED MATCHING, and it can also be considered a related problem to PMS. Stockmeyer and Vazirani \cite{SV82} discussed  INDUCED MATCHING under the 'Risk-free' marriage problem. Since then INDUCED MATCHING is extensively studied. Hardness and polynomial time solvable results have been obtained with various structural assumptions \cite{IMS1,IMS2,IMS3,IMS4,IMS5,IMS6,IMS7,IMS8}. Exact algorithms \cite{IMEX1,IMEX2}, and FPT and kernelization results \cite{IMP1,IMP2,IMP3} have also been obtained.

\section{Preliminaries}

We use $[n]$ to denote the set $\{1,2,....,n\}$, and $[0]$ denotes an empty set. For a function $f:X\to Y$, for an element $e\in Y$,  $f^{-1}(e)$ is defined to be the set of all elements of $X$ that map to $e$. Formally, $f^{-1}(e)=\{x\mid x\in X \land f(x)=e\} $. 

\subsection{Graph Notations} All the graphs that we refer to are simple and finite. We mostly use standard notations and terminologies. We use $G=(V,E)$ to denote a graph with vertex set $V$ and edge set $E$. $E(G)$ denotes the set of edges of graph $G$, and $V(G)$ denotes the set of vertices of $G$. For $E'\subseteq E$, $V(E')$ denotes the set of all vertices of $G$ with at least one edge in $E'$ incident on it. For a vertex set $X\subseteq V$, $G[X]$ denotes the induced subgraph of $G$ on vertex set $X$, and $G-X$ denotes the graph $G[V\setminus X]$.
 For an edge set $E'\subseteq E$, $G[E']$ denotes the subgraph of $G$ on edge set $E'$ i.e. $G[E']=(V(E'),E')$.

For disjoint vertex sets $A\subseteq V$ and $B \subseteq V$, $E(A,B)$ denotes the set of all the edges of $G$ with one endpoint in $A$ and another in $B$. For a vertex $v\in V$, we use $N(v)$ to denote the open neighborhood of $v$, i.e. set of all vertices adjacent to $v$ in $G$. We use $N[v]$ to denote the closed neighborhood of $v$, i.e. $N(v)\cup \{v\}$. 

A graph is a \textit{cluster graph} if it is a vertex disjoint union of cliques. The maximal cliques of a cluster graph are called cliques or clusters. A graph is a \textit{co-cluster graph} if it is a complement of a cluster graph or equivalently a complete multipartite graph.\\

\subsection{Parameterized Complexity} 

For details on parameterized complexity, we refer to \cite{DBLP:books/sp/CyganFKLMPPS15,DBLP:series/txcs/DowneyF13}, and recall some definitions here.

\begin{definition}[\cite{DBLP:books/sp/CyganFKLMPPS15}]
A \textit{parameterized problem} is a language $L \subseteq \Sigma^* \times \mathbb{N} $ where $\Sigma$ is a fixed and finite alphabet. For an instance $I=(x,k) \in \Sigma^* \times \mathbb{N} $, $k$ is called the parameter.  A parameterized problem is called \textit{fixed-parameter tractable} (FPT) if there exists an algorithm $\cal A$ (called a \textit{fixed-parameter algorithm} ), a computable function $f:\mathbb{N} \to \mathbb{N}$, and a constant $c$ such that, the algorithm $\cal A$ correctly decides whether $(x,k)\in L$ in time bounded by $f(k).|(x,k)|^c$. The complexity class containing all fixed-parameter tractable problems is called FPT.
\end{definition}
In the above definition, $|(x,k)|$ denotes the size of the input for a problem instance $(x,k)$.
Informally, a $W[1]$-hard problem is unlikely to be fixed parameter tractable, see \cite{DBLP:books/sp/CyganFKLMPPS15} for details on complexity class $W[1]$.
\begin{definition}[\cite{DBLP:books/sp/CyganFKLMPPS15}]
Let $P,Q$ be two parameterized problems. A parameterized reduction from $P$ to $Q$ is an algorithm which for an instance $(x,k)$ of $P$ outputs an instance $(x',k')$ of $Q$ such that:
\begin{itemize}
    \item $(x,k)$ is yes instance of $P$ if and only if $(x',k')$ is a yes instance of $Q$,
    \item $k'\leq g(k)$ for some computable function $g$, and
    \item the reduction algorithm takes time $f(k)\cdot |x|^{O(1)}$ for some computable function $f$
\end{itemize}
\end{definition}

\begin{theorem}[\cite{DBLP:books/sp/CyganFKLMPPS15}]\label{def-parameterized reduction}
If there is a parameterized reduction from $P$ to $Q$ and $Q$ is fixed parameter tractable then $P$ is also fixed parameter tractable.
\end{theorem}

For details on \textit{kernelization}, we refer to \cite{DBLP:books/sp/CyganFKLMPPS15} and recall the basic definition of a kernel here. A \textit{kernel} for a parameterized problem $P$ is an algorithm $A$ that given an instance $(x,k)$ of $P$ takes polynomial time and outputs an instance $(x',k')$ of $P$, such that (i) $(x,k)$ is a yes instance of $P$ if and only if $(x',k')$ is a yes instance of $P$, (ii) $|x'|+k' \leq f(k)$ for some computable function $k$. If $f(k)$ is polynomial in $k$, then we call it a polynomial kernel.

\section{Parameterized lower bounds}\label{W1-hard}

\subsection{W[1]-Hardness for Split Graphs }\label{section-whard-split}

In this section, we will prove the following theorem.
\begin{theorem}\label{split_hardness}
PMS is W[1]-hard for split graphs when parameterized by solution size $k$.
\end{theorem}

IRREDUNDANT SET is known to be $W[1]$-complete when parameterized by the number of vertices in the set  \cite{DOWNEY2000155}. We will give a parameterized reduction from IRREDUNDANT SET to PMS with solution size as the parameter. We also note that our construction in the reduction is similar to the one given in \cite{IMP1}.
\begin{definition}[\cite{DOWNEY2000155}]
 A set of vertices $I\subseteq V$ in a graph $G=(V,E)$ is called irredundant, if every vertex $v\in I$ has a private closed neighbor $p(v)$ in $V$ satisfying the following conditions:
 \begin{enumerate}
     \item $v=p(v)$ or $p(v)$ is adjacent to $v$,
     \item no other vertex in $I$ is adjacent to $p(v)$.
 \end{enumerate}
\end{definition}

 \begin{figure}[H]
   
    \centering
    \begin{tikzpicture}

    \node [ellipse,draw=black!10, minimum height=2cm,minimum width= 5cm, label ={180:$Y$}] (z) at (1.5,-0.8) {};

    \foreach \i in {1,2,...,3}{
     
            \node[shape=circle, draw, fill = black, scale= 0.4, font=\footnotesize,label={220:\scriptsize{$y_\i$}}] ({\i}) at (-0.5+\i/2,-1){};
   }
   \foreach \i in {4,5,...,6}{
     
            \node[shape=circle, draw, fill = black, scale= 0.4, font=\footnotesize,] ({\i}) at (-0.5+\i/2,-1){};
   }
    \foreach \i in {7}{
     
            \node[shape=circle, draw, fill = black, scale= 0.4, font=\footnotesize,label={0:\scriptsize{$y_n$}}] ({\i}) at (-0.5+\i/2,-1){};
   }
    
    \foreach \i in {3,...,7}{
     
           \draw[thin,black!20] (1) to[out=50,in=130] (\i);
   }
    \foreach \i in {4,...,7}{
     
           \draw[thin,black!20] (2) to[out=50,in=130] (\i);
   }
    \foreach \i in {5,...,7}{
     
           \draw[thin,black!20] (3) to[out=50,in=130] (\i);
   }
    \foreach \i in {6,...,7}{
     
           \draw[thin,black!20] (4) to[out=50,in=130] (\i);
   }
   \foreach \i in {7}{
     
           \draw[thin,black!20] (5) to[out=50,in=130] (\i);
   }


    \foreach \i in {1,2,...,3}{
     
            \node[shape=circle, draw, fill = black, scale= 0.4, font=\footnotesize,label={-10:\scriptsize{$z_\i$}}] ({\i+0}) at (-2.5+\i,-4){};
   }
   \foreach \i in {4,5,...,6}{
     
            \node[shape=circle, draw, fill = black, scale= 0.4, font=\footnotesize,] ({\i+0}) at (-2.5+\i,-4){};
   }
    \foreach \i in {7}{
     
            \node[shape=circle, draw, fill = black, scale= 0.4, font=\footnotesize,label={-10:\scriptsize{$z_n$}}] ({\i+0}) at (-2.5+\i,-4){};
   }
   
    \draw[thin,black!30] (1)--(2)--(3)--(4)--(5)--(6)--(7);
    \foreach \i in {1,2,...,7}{
     
            \draw[thin,black!40] (\i)--(\i+0);

   }
   \draw[thin,black!50] (5)-- (1+0) -- (3) (1+0)--(7);
   \draw[thin,black!50] (5)-- (3+0) -- (1) (3+0)--(7);
   \draw[thin,black!50] (5)-- (1+0) -- (3) (1+0)--(7);
   \draw[thin,black!50] (3)-- (5+0) -- (1) (5+0)--(7);
   \draw[thin,black!50] (5)-- (7+0) -- (3) (7+0)--(1);

  \draw[thin,black!50] (4)-- (2+0) -- (6);
  \draw[thin,black!50] (2)-- (4+0) -- (6);
  \draw[thin,black!50] (4)-- (6+0) -- (2);
  
  \node [ellipse,draw=black!20, minimum height=1.4cm,minimum width= 9cm, label ={180:$Z$}] (z) at (1.5,-4) {};

\end{tikzpicture}
\caption{reduction from IRREDUNDANT SET to PMS, vertex set $Y$ forms a clique and $Z$ forms an independent set. A vertex $z_i$ is connected to a vertex $y_j\in Y$ if $v_j\in N[v_i]$ in input graph $G$.}
\end{figure}
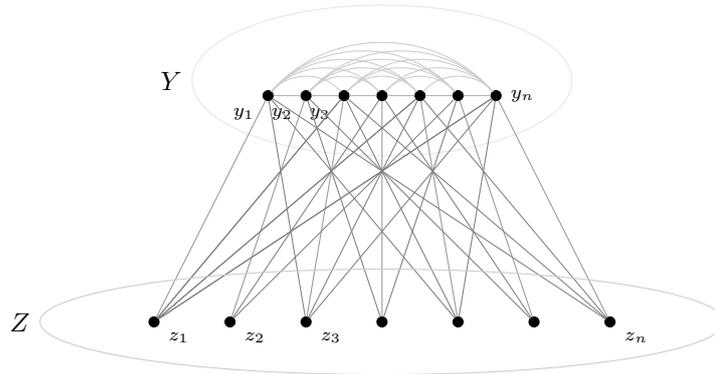

Let $(G=(V,E),k)$ be an instance of IRREDUNDANT SET and let $V=\{v_1,v_2,...,v_n\}$.
We construct a new graph $G'$ as follows.
\begin{itemize}
    \item Create a vertex set $Y=\{y_i\mid v_i\in V\}$ and connect every vertex of $y_i\in Y$ to $y_j\in Y$ when $y_i\neq y_j$, that is $Y$ forms a clique, and let $E_y$ be the set of these edges.
    \item Create a vertex set $Z=\{z_i\mid v_i\in V\}$ and connect every vertex $z_i\in Z$ to $y_j\in Y$ if $v_j\in N[v_i]$ in input graph $G$, let $E'$ be the set of these edges. 
\end{itemize}

We define $G'=(Y\cup Z, E_y\cup E')$, which can be constructed in time polynomial in $|V|$, and $G'$ is a split graph.

\begin{proposition}\label{B'inV'}
$G$ has an irredundant set of size $k$ if and only if $G'$ has a pair of perfectly matched sets of size $k$.
\end{proposition}
\begin{proof}
For the forward direction, let $I\subseteq V$ be an irredundant set of size $k$ in $G$. For every vertex $v_i\in I$, pick exactly one private closed neighbor $v_{i'}$ of $v_i$ in $G$ and let $J$ be the set of these picked vertices. Let
$I_a= \{y_{i}\mid y_{i}\in Y\land v_{i}\in I \}$, and let
$J_b= \{z_{i'}\mid z_{i'}\in Z \land v_{i'}\in J\}$. We claim that $(I_a,J_b)$ is a pair of perfectly matched sets in $G'$. Since for every $v_i\in I$, $v_{i'}$ is a private closed neighbor of $v_i$ in $G$, by construction of $G'$, every $z_{i'}\in I_b$ is adjacent to only $y_i$ in $I_a$, and every $y_{i}\in I_a$ is adjacent to only $z_{i'}$ in $I_b$.

For the other direction, let $(A,B)$ be a pair of perfectly matched sets in $G'$ of size $k$. If $k = 1$, then let $y_i$ or $z_j$ be the only vertex in $A$, then $\{v_i\}$ or $\{v_j\}$ is an irredundant set of size one in $G$. If $k=|E(A,B)|\geq 2$, then either $A\subseteq Y$ and $B\subseteq Z$, or $A\subseteq Z$ and $B\subseteq Y$, this is due to the fact that $Y$ forms a clique and $Z$ forms an independent set in $G'$, and if both $A\cap Y$ and $B\cap Y$ are non-empty, then $|E(A,B)|$ must be 1. Let us assume that $A\subseteq Y$ and $B\subseteq Z$. In this case, let $I^*=\{v_i\mid v_i\in V\land y_i\in A\}$. Since $(A,B)$ is a pair of perfectly matched sets in $G'$, for every $y_i\in A$, there must be a vertex $z_{i'}$ in $B$ such that $y_i$ is the only neighbor of $z_{i'}$ in $A$. Therefore, every $v_i\in I^*$ has a private closed neighbor $v_{i'}$ in $G$, and $I^*$ is an irredundant set of size $k$ in $G$. This finishes the proof.
\end{proof}

\subsection{ W[1]-Hardness for Bipartite Graph }\label{section-whard-bipart}

In this section, we prove the following theorem.
\begin{theorem}\label{bi_hardness}
PMS is W[1]-hard for bipartite graphs when parameterized by solution size $k$.
\end{theorem}

We give a reduction from PMS on general graphs to PMS on bipartite graphs.

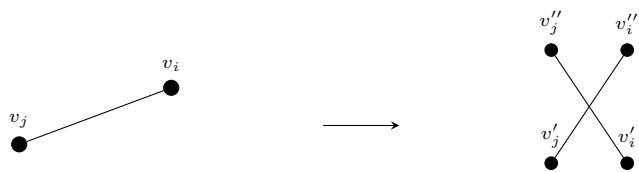
\begin{figure}[H]
    \centering
    
    \begin{tikzpicture}[]
\node [shape=circle, draw, fill=black,scale =0.6, label={\scriptsize{$v_i$}}] (v) at (0,1) {};
\node [shape=circle, draw, fill=black,scale =0.6, label={\scriptsize{$v_j$}}] (u) at (-2,0.25) {};

\draw (v)--(u);

\node [shape=circle, draw, fill=black,scale =0.5, label={\scriptsize{$v'_i$}}] (v1) at (6,0) {};
\node [shape=circle, draw, fill=black,scale =0.5, label={\scriptsize{$v''_i$}}] (v1') at (6,1.5) {};

\node [shape=circle, draw, fill=black,scale =0.5, label={\scriptsize{$v'_j$}}] (u1) at (5,0) {};
\node [shape=circle, draw, fill=black,scale =0.5, label={\scriptsize{$v''_j$}}] (u1') at (5,1.5) {};

\draw (v1)--(u1');
\draw (u1)--(v1');

\draw [->,>=stealth] (2,0.5) -- (3,0.5);

\end{tikzpicture}

    \caption{For every vertex $v_i$ create two  vertices $v'_i$ and $v''_i$, and connect $v'_i$ to $v''_j$ if $v_j\in N(v_i)$.}
    \label{fig:fig_bip_red}
\end{figure}

Let $(G=(V,E),k)$ be an instance PMS on general graphs, let $V=\{v_1,v_2,...,v_n\}$. We construct $G'$ as follows.
\begin{itemize}
    \item Create two copies of $V$ and call it $V'$ and $V''$, we refer copy of a vertex $v_i\in V$ in $V'$ as $v'_i$ and in $V''$ as $v''_i$.
    \item Connect $v'_i$ to $v''_j$ if $v_i$ is adjacent to $v_j$ in $G$. Let $E'$ be the set of all these edges.
\end{itemize}
We define $G'=(V'\cup V'',E')$, which can be constructed in time polynomial in $|V|$.
\begin{proposition}\label{bi_hard_correctness}
 $G$ has a pair of perfectly matched sets of size $k$ if and only if $G'$ has a pair of perfectly matched sets of size $2k$.
\end{proposition}
\begin{proof}
For the forward direction, let $(A,B)$ be a pair of perfectly matched sets in $G$ such that $|E(A,B)|=k$. Let $A'=\{v'_i\ |\ v_i\in A\}$, $A''=\{v''_i\ |\ v_i\in A\}$, $B'=\{v'_i\ |\ v_i\in B\}$, and $B''=\{v''_i\ |\ v_i\in B\}$. Observe that both $(A',B'')$ and $(A'',B')$ are perfectly matched sets of $G'$. Further, $(A''\cup A', B''\cup B')$ are perfectly matched sets of $G'$ as no vertex in $A''$ has a neighbor in $B''$ in $G'$ and similarly no vertex in $A''$ has a neighbor in $B''$ in $G'$. We further have $ |E(A''\cup A', B''\cup B')| =2k$.

For the other direction, let $(A,B)$ be perfectly matched sets in $G'$ such that  $|E(A,B)|= 2k$.  Due to the construction of $G'$, any vertex in $V'$ can only be matched to a vertex in $V''$ and vice versa. Thus there are $2k$ vertices from $V'$ in $A\cup B$. Then, either $|A\cap V'|\geq k$  or $|B\cap V'|\geq k$. W.l.o.g let $|A\cap V'|\geq k$, let $A'=A\cap V'$ and $B''\subseteq B$ be the vertices that are matched to $A'$ in $(A,B)$, clearly $B''\subseteq V''$ due to the construction. Let $A^*=\{ v\ |v' \in A'\}$ and $B^*=\{ v\ |v'' \in B''\}$; then $(A^*, B^*)$ is a pair of perfectly matched sets in $G$ such that $|E(A^*, B^*)|\geq k$.
\end{proof}

\section{Kernelization Lower Bounds} \label{section-kernel-lb}

We refer to the work of Bodlaender, Thomassé and Yeo \cite{BODLAENDER_PPT} for the details on polynomial time parameter transformation, and recall its definition here.

\begin{definition}[\cite{BODLAENDER_PPT}]
For two parameterized problems $P$ and $Q$, we say that there exists a polynomial time parameter transformation (ppt) from $P$ to $Q$, denoted by $P\leq_{ppt }Q$, if there exists a polynomial time computable function $f: \{0,1\}^* \times \mathbb{N} \to \{0,1\}^* \times \mathbb{N}$, and a polynomial $p: \mathbb{N}\to \mathbb{N}$, and for all $x\in \{0,1\}^*$ and $k\in \mathbb{N}$, if $f((x,k))=(x',k')$, then the following hold.
\begin{itemize}
    \item $(x,k)\in P$ if and only if $(x',k')\in Q$, and
    \item $k'\leq p(k)$.
\end{itemize}
Here, $f$ is called the polynomial time parameter transformation.
\end{definition}

\begin{theorem}[\cite{BODLAENDER_PPT}]
For two parameterized problems $P$ and $Q$, let $P'$ and $Q'$ be their derived classical problems. Suppose that $P'$ is NP-complete, and $Q'\in \text{NP}$. Suppose that $f$ is a polynomial time parameter transformation from $P$ to $Q$. Then, if $Q$ has a polynomial kernel, then $P$ has a polynomial kernel.
\end{theorem}

In the remaining part of this section, we will prove the following theorem.

\begin{theorem}\label{theorem-kernel}
PMS does not admit a polynomial kernel parameterized by vertex cover number unless $\text{NP}\subseteq \text{coNP/poly}$.
\end{theorem}

 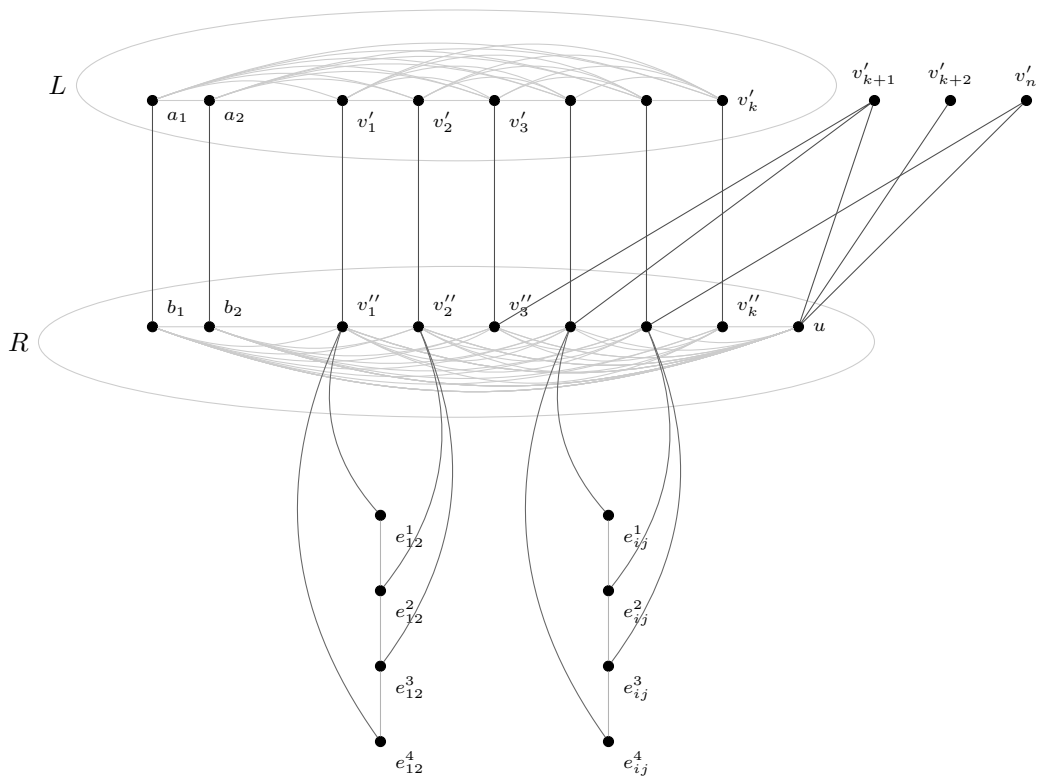
\begin{figure}[H]
   
    \centering
    \begin{tikzpicture}

    \node [ellipse,draw=black!20, minimum height=2cm,minimum width= 10cm, label ={180:$L$}] (z) at (2,-0.8) {};
   \node [ellipse,draw=black!20, minimum height=2cm,minimum width= 11cm, label ={180:$R$}] (z) at (2,-4.2) {};
   

    \foreach \i in {1,2,...,3}{
     
            \node[shape=circle, draw, fill = black, scale= 0.4, font=\footnotesize,label={-20:\scriptsize{$v'_\i$}}] ({\i}) at (-0.5+\i,-1){};
   }
   \foreach \i in {4,5}{
     
            \node[shape=circle, draw, fill = black, scale= 0.4, font=\footnotesize,] ({\i}) at (-0.5+\i,-1){};
   }
    \foreach \i in {6}{
     
            \node[shape=circle, draw, fill = black, scale= 0.4, font=\footnotesize,label={0:\scriptsize{$v'_k$}}] ({\i}) at (-0.5+\i,-1){};
   }
    
    \node[shape=circle, draw, fill = black, scale= 0.4, font=\footnotesize,label={-20:\scriptsize{$a_2$}}] ({a2}) at (-1.25,-1){};
    \node[shape=circle, draw, fill = black, scale= 0.4, font=\footnotesize,label={-20:\scriptsize{$a_1$}}] ({a1}) at (-2,-1){};
    
        \node[shape=circle, draw, fill = black, scale= 0.4, font=\footnotesize,label={20:\scriptsize{$b_2$}}] ({b2}) at (-1.25,-4){};
    \node[shape=circle, draw, fill = black, scale= 0.4, font=\footnotesize,label={20:\scriptsize{$b_1$}}] ({b1}) at (-2,-4){};
   
   
   \foreach \i in {3,...,6}{
           \draw[thin,black!20] (1) to[out=30,in=150] (\i);
   }
   \foreach \i in {4,...,6}{
           \draw[thin,black!20] (2) to[out=30,in=150] (\i);
   }
   \foreach \i in {5,...,6}{
           \draw[thin,black!20] (3) to[out=30,in=150] (\i);
   }
    
   \foreach \i in {6}{
           \draw[thin,black!20] (4) to[out=30,in=150] (\i);
   }
    \foreach \i in {1,...,6}{
           \draw[thin,black!20] (a1) to[out=20,in=160] (\i);
   }
      \foreach \i in {2,...,6}{
           \draw[thin,black!20] (a2) to[out=20,in=160] (\i);
   }

    \foreach \i in {1,2,...,3}{
     
            \node[shape=circle, draw, fill = black, scale= 0.4, font=\footnotesize,label={20:\scriptsize{$v''_\i$}}] ({\i+0}) at (-0.5+\i,-4){};
   }
   \foreach \i in {4,5}{
     
            \node[shape=circle, draw, fill = black, scale= 0.4, font=\footnotesize,] ({\i+0}) at (-0.5+\i,-4){};
   }
    \foreach \i in {6}{
     
            \node[shape=circle, draw, fill = black, scale= 0.4, font=\footnotesize,label={20:\scriptsize{$v''_k$}}] ({\i+0}) at (-0.5+\i,-4){};
   }
   
   \node[shape=circle, draw, fill = black, scale= 0.4, font=\footnotesize,label={0:\scriptsize{$u$}}] (u) at (6.5,-4){};


   \foreach \i in {3,...,6}{
           \draw[thin,black!20] (1+0) to[out=-30,in=-150] (\i+0);
           \draw[thin,black!20] (1+0) to[out=-20,in=-160] (u);
   }
   \foreach \i in {4,...,6}{
           \draw[thin,black!20] (2+0) to[out=-30,in=-150] (\i+0);
           \draw[thin,black!20] (2+0) to[out=-20,in=-160] (u);
   }
   \foreach \i in {5,...,6}{
           \draw[thin,black!20] (3+0) to[out=-30,in=-150] (\i+0);
           \draw[thin,black!20] (3+0) to[out=-20,in=-160] (u);
   }
    
   \foreach \i in {6}{
           \draw[thin,black!20] (4+0) to[out=-30,in=-150] (\i+0);
           \draw[thin,black!20] (4+0) to[out=-20,in=-160] (u);
            \draw[thin,black!20] (5+0) to[out=-20,in=-160] (u);
   }
    \foreach \i in {1,...,6}{
           \draw[thin,black!20] (b1) to[out=-20,in=-160] (\i+0);
           \draw[thin,black!20] (b1) to[out=-20,in=-160] (u);
   }
      \foreach \i in {2,...,6}{
           \draw[thin,black!20] (b2) to[out=-20,in=-160] (\i+0);
           \draw[thin,black!20] (b2) to[out=-20,in=-160] (u);
   }


    \node[shape=circle, draw, fill = black, scale= 0.4, font=\footnotesize,label={\scriptsize{$v'_{k+1}$}}] ({u1}) at (7.5,-1){};
    \node[shape=circle, draw, fill = black, scale= 0.4, font=\footnotesize,label={\scriptsize{$v'_{k+2}$}}] ({u2}) at (8.5,-1){};
    \node[shape=circle, draw, fill = black, scale= 0.4, font=\footnotesize,label={\scriptsize{$v'_n$}}] ({u3}) at (9.5,-1){};
    
     \draw[thin,black!70] (u)--(u1) (u)--(u2) (u)--(u3);
     
      \draw[thin,black!70] (3+0)--(u1) (4+0)--(u1) (5+0)--(u3);

   
    \draw[thin,black!20] (a1)--(a2)--(1)--(2)--(3)--(4)--(5)--(6);
     \draw[thin,black!20] (b1)--(b2)--(1+0)--(2+0)--(3+0)--(4+0)--(5+0)--(6+0)--(u);
     
      \draw[thin,black!70] (a1)--(b1)  (a2)--(b2);
    \foreach \i in {1,2,...,6}{
     
            \draw[thin,black!70] (\i)--(\i+0);

   }
   
   \foreach \i in {1,2,...,4}{
     
            \node[shape=circle, draw, fill = black, scale= 0.4, font=\footnotesize,label={-10:\scriptsize{$e_{12}^{\i}$}}] ({e12\i}) at (1,-5.5-\i){};
   }

  \draw[thin,black!30] (e121)--(e122)--(e123)--(e124);
 
  \draw[thin,black!60,bend right] (1+0) to ({e121});
  \draw[thin,black!60,bend right] (1+0) to ({e124});
  
   \draw[thin,black!60,bend left] (2+0) to ({e122});
  \draw[thin,black!60,bend left] (2+0) to ({e123});

   \foreach \i in {1,2,...,4}{
     
            \node[shape=circle, draw, fill = black, scale= 0.4, font=\footnotesize,label={-10:\scriptsize{$e_{ij}^{\i}$}}] ({e45\i}) at (4,-5.5-\i){};
   }

  \draw[thin,black!30] (e451)--(e452)--(e453)--(e454);
 
  \draw[thin,black!60,bend right] (4+0) to ({e451});
  \draw[thin,black!60,bend right] (4+0) to ({e454});
  
   \draw[thin,black!60,bend left] (5+0) to ({e452});
  \draw[thin,black!60,bend left] (5+0) to ({e453});

\end{tikzpicture}
\caption{Reduction from CLIQUE to PMS, vertex set L and R form cliques. A vertex $v'_{k+i}$ is connected to a vertex $v''_j\in R$ if $v_j$ is not connected to $v_{k+i}$ in input graph $G$.}
\end{figure}

One can observe that $\text{PMS} \in \text{NP}$, as the solution certificate $(A,B)$ can be easily verified in polynomial time. Further, CLIQUE (an NP-complete problem) does not admit a polynomial kernel when parameterized by the size of a vertex cover unless $\text{NP}\subseteq \text{coNP/poly}$ \cite{Bodlaender_cross}.
Thus, it will suffice to obtain a polynomial time parameter transformation from CLIQUE to PMS with parameter vertex cover number. 

To this end, let $(G,l,V_c)$ be an instance of CLIQUE, where we need to decide if the input graph $G$ with vertex cover $V_c$ contains a clique of size $l$, here the parameter is the size of vertex cover $|V_c|\leq k$. Note that we are considering a vertex cover of $G$ of size $k$ to be a part of the input. However we are not dependent on this assumption, as there exists a well known polynomial time algorithm that finds a $2$-approximation of a minimum vertex cover.

For notational simplicity, let $V=\{v_1,v_2,...,v_n\}$ be the vertex set of $G$ and let $V_c=\{v_1,v_2,...,v_k\}$ be a vertex cover of $G$. The rest of the vertices of $G$ are $\{v_{k+1},v_{k+2},...,v_n\}$. Let $\bar{E}$ be the set of all the non-edges in $G$, i.e. $\bar{E}=\{v_iv_j\ |\ v_i\neq v_j \land v_iv_j\not \in E(G)\}$. Further, let $\bar{E}(V_c)$ be all the non-edges with both the endpoints in $V_c$.
We now construct a graph $G'$ as follows.
\vskip 0.1cm
\begin{itemize}
    \item Create a vertex set $L=\{v'_i|\ v_i\in V_c\}\cup \{a_1,a_2\}$, and  for every distinct $x,y\in L$, connect $x$ and $y$ to each other, this way $L$ forms a clique.
    \item Create a vertex set $R=\{v''_i|\ v_i\in V_c\}\cup \{b_1,b_2,u\}$, and  for every distinct $x,y\in R$, connect $x$ and $y$ to each other, this way $R$ forms a clique.
    \item Connect $a_1$ to $b_1$, $a_2$ to $b_2$, and $v'_{i}$ to $v''_{i}$ for every $i\in [k]$.
     \item Create a vertex set $F= \{v'_{k+i}|\ v_{k+i}\in V(G)\setminus V_c\}$, and connect every vertex of $F$ to $u$.
     \vskip 0.1cm
     \item Connect a vertex $v'_{k+i}\in F$ to $v''_j\in R$ if $v_{k+i}$ and ${v_j}$ are not connected in $G$.
     \vskip 0.1cm
     \item Create a vertex set $X_{\bar{e}}= \{e_{ij}^1,e_{ij}^2,e_{ij}^3,e_{ij}^4|\ v_iv_j\in \bar{E}(V_c) \land (i<j)\}$. Further, for every non-edge $v_iv_j\in \bar{E}(V_c)$ where $i<j$, connect $e_{ij}^1$ to $e_{ij}^2$, $e_{ij}^2$ to $e_{ij}^3$, and $e_{ij}^3$ to $e_{ij}^4$, that is we create a path on 4 vertices, and connect $v''_{i}$ to $e_{ij}^1$ and $e_{ij}^4$, and connect $v''_j$ to $e_{ij}^2$ and $e_{ij}^3$.
\end{itemize}

Observe that the construction can be achieved in time polynomial in $|V|$. Further, the vertex set $V(G')\setminus F$ is a vertex cover of $G'$, which is of size at most $|L|$ + $|R|$ + $4\cdot |\bar{E}(V_c)|$ that is $O(k^2)$. In the following proposition, we state and then prove the correctness of the reduction.

\begin{proposition}\label{popo:kernel-correctness}
$G$ has a clique of size $l$ if and only if $G'$ has a pair of perfectly matched sets of size $2+2\cdot |\bar{E}(V_c)|+l$.
\end{proposition}
\begin{proof}
For the forward direction, let $C$ be a clique of size $l$ in $G$. Recall that the size of a clique in $G$ can be at most $1$ more than the size of a vertex cover of $G$, thus $l$ can be at most $k+1$. We construct two sets $A$ and $B$ as follows, for every vertex $v_i$ in $C$ which also belongs to $V_c$, we put $v'_i$ in $A$ and $v''_i$ in $B$, if there is a vertex $v_{k+j}$ in $C$ which is not in $V_c$, we put $v'_{k+j}$ in $A$ and $u$ in $B$, observe that at most $1$ such vertex can be in $C$. Further, for every non-edge $v_iv_j\in \bar{E}(V_c)$ such that $i<j$, if $v_i\in C$ (then certainly $v_j\not \in C$), then put $e_{ij}^1$ and $e_{ij}^4$ in $B$ and $e_{ij}^2$ and $e_{ij}^3$ in $A$, else if $v_i\not \in C$, then put $e_{ij}^1$ and $e_{ij}^4$ in $A$, and  $e_{ij}^2$ and $e_{ij}^3$ in $B$. Lastly we put $a_1$ and $a_2$ in $A$ and $b_1$ and $b_2$ in $B$. A direct check can verify that $(A,B)$ is a pair of perfectly matched sets in $G'$ and $|E(A,B)|= 2+2\cdot |\bar{E}(V_c)|+l$.

For the other direction, let $(A,B)$ be a pair of perfectly matched sets in $G'$ such that $|E(A,B)|= 2+2\cdot |\bar{E}(V_c)|+l$. For the proof of this direction, we will modify the set $A$ and/or $B$ to obtain desired vertices in each set while maintaining that $(A,B)$ remains a pair of perfectly matched sets and that $|E(A,B)|$ does not decrease.

For a pair $(A,B)$ of perfectly matched sets, we say a vertex $x\in A$ (resp. $B$) is \textit{matched} to $y\in B$ (resp. $A$) if $x$ and $y$ are neighbors. We will sequentially apply the following modifications, after verifying if the modification is applicable. Precedence of every modification is in the same order in which they are described.

\begin{itemize}
    \item \textbf{M1:} If there exist two distinct vertices $x,y\in L$ such that $x$ is in $A$ and $y$ is in $B$, as well as there exist two distinct vertices $x',y'\in R$ such that $x'$ is in $A$ and $y'$ is in $B$. Then we set $A=(A\setminus\{x,x'\}) \cup \{a_1,a_2\}$ and $B= (B\setminus \{y,y'\}) \cup \{b_1,b_2\}$.
    Observe that $(A,B)$ remains a pair of perfectly matched sets and $|E(A,B)|$ remains unchanged.
     This is due to the fact that both $L$ and $R$ form cliques and thus no other vertices from $L$ except $x,y$ and no other vertices from $R$ except $x',y'$ could belong to $A$ or $B$.
    
     \item  \textbf{M2:} If there exist two distinct vertices $x,y\in L$ such that $x$ is in $A$ and $y$ is in $B$, and either $ A\cap R = \emptyset$ or $B\cap R = \emptyset$. In this case, neither $A\setminus\{x\}$ nor $B\setminus\{y\}$ contains any vertex from $L$ (as $L$ is a clique). Further, if $A\cap R=\emptyset$, then $b_1$ cannot belong to $B$, as it cannot be matched to any vertex in $A$ ($x$ and $y$ matched to each other). In this case, we set $A=(A\setminus\{x\}) \cup \{a_1\}$ and $B=(B\setminus \{y\}) \cup \{b_1\}$. Similarly, if  $B\cap R=\emptyset$, then $b_1$ cannot belong to $A$ as it cannot be matched to any vertex in $B$. In this case, we set $A=(A\setminus\{x\}) \cup \{b_1\}$ and $B=(B\setminus \{y\}) \cup \{a_1\}$. In both the cases, $(A,B)$ remains a pair of perfectly matched sets and $|E(A,B)|$ remains unchanged.
     
      \item   \textbf{M3:} If there exist two distinct vertices $x,y\in R$ such that $x$ is in $A$ and $y$ is in $B$, and either $ A\cap L = \emptyset$ or $B\cap L = \emptyset$. In this case, neither $A\setminus\{x\}$ nor $B\setminus\{y\}$ contains any vertex from $R$ (as $R$ is a clique). Further, if $A\cap L=\emptyset$, then $a_1$ cannot belong to $B$ as it cannot be matched to any vertex in $L$ ($x$ and $y$ matched to each other). In this case, we set $A= (A\setminus\{x\}) \cup \{b_1\}$ and $B= (B\setminus \{y\}) \cup \{a_1\}$. Similarly, if $B\cap L=\emptyset$, then $a_1$ cannot belong to $A$ as it cannot be matched to any vertex in $B$. In this case, we set $A= (A\setminus\{x\}) \cup \{a_1\}$ and $B= (B\setminus \{y\}) \cup \{b_1\}$.
      In both the cases, $(A,B)$ remains a pair of perfectly matched sets and $|E(A,B)|$ remains unchanged.
\end{itemize}

After applying the above modifications, we may assume that for a pair of perfectly matched sets  $(A,B)$, either $L\cap B = \emptyset$ and $R\cap A = \emptyset$, or $R\cap B = \emptyset$ and $L\cap A = \emptyset$. For the simplicity of arguments, we assume that $L\cap B = \emptyset$ and $R\cap A = \emptyset$ and proceed as follows:
\begin{itemize}
    \item  \textbf{M4:} For every vertex $v''_i\in (R\cap B)\setminus \{u\}$, if the only neighbor of $v''_i$ in $A$ is $x$ and $x\neq v'_i$, then we remove $x$ from $A$ and add $v'_i$ to $A$. Observe that it is safe to do so as $v'_i$ is adjacent to only $v''_i$ outside $L$ and $L$ is disjoint from $B$. Also this modification does not change the size of $E(A,B)$.
\end{itemize}

To this end, after applying the above modification exhaustively, we may also assume that in $(A,B)$, every vertex of $(R\cap B) \setminus \{u\}$ is matched by a vertex in $L$. Observe that if two distinct vertices $v''_i,v''_j\in B$ are such that $v_i$ and $v_j$ is not connected in $G$ and $v''_i$ is matched to $v'_i$ and $v''_j$ is matched to $v'_j$ in $(A,B)$, then none of the vertices from $\{e^1_{ij},e^2_{ij},e^3_{ij},e^4_{ij}\}$ can belong to $A$ without violating property of perfectly matched sets, and they can not be matched by $v''_i$ or $v''_j$, and hence none of them belongs to either $A$ or $B$. Thus, we modify $(A,B)$ as follows:

\begin{itemize}
    \item \textbf{M5:} If two distinct vertices $v''_i,v''_j\in B$ are such that $v''_i$ is matched to $v'_i$ and $v''_j$ is matched to $v'_j$ in $(A,B)$ and $v_i$ and $v_j$ are not connected in $G$, then we remove  $v''_i,v''_j$ from $B$ and remove $v'_i,v'_j$ from $A$, we then put $e^2_{ij},e^3_{ij}$ in $B$ and $e^1_{ij},e^4_{ij}$ in $A$. Observe that this modification maintains that $(A,B)$ remains a pair of perfectly matched sets and $|E(A,B)|$ remains unchanged.
\end{itemize}

Recall that $|E(A,B)|$ is $2+2\cdot \bar{E}(V_c)+l$. Applying the above modification exhaustively, we also ensure that none of the vertices in $R\cap B$ is matched to a vertex from $X_{\bar{e}}$. Thus, every vertex of $X_{\bar{e}}\cap (A\cup B)$ must be matched to a vertex of $X_{\bar{e}}\cap (A\cup B)$ in $(A,B)$.  Since there are at most $4\cdot |\bar{E}(V_c)|$ such vertices, they contribute at most $2\cdot |\bar{E}(V_c)|$ to  $|E(A,B)|$. If we consider $a_1,a_2$ and $b_1,b_2$ to be part of $A$ and $B$ respectively, this leaves us with remaining $l$ edges in $E(A,B)$, the endpoints of these edges which belong to set $B$ must be from $R\setminus \{b_1,b_2\}$, let $C$ be the set of these $l$ vertices. If $C$ contains $u$, then let $v'_{k+p}$ be the only neighbor of $u$ in $A$. We define $C'=\{ v_i | \ v''_i\in C\}\cup \{v_{k+p}\}$ if $C$ contains $u$, otherwise $C'=\{ v_i | \ v''_i\in C\}$. Observe that $|C'|=l$. We claim that $C'$ is a clique in $G$. To prove the claim, recall modification M5, which ensures that every distinct $v_i,v_j\in C'\cap V_c$ are connected, and if there is a vertex $v_{k+p}$ outside $V_c$ in $C'$, $v_{k+P}$ is connected to every $V_c\cap C'$ in $G$, otherwise recalling construction of $G'$, $v'_{k+P}$ would be connected to at least one vertex in $C\setminus \{u\}$ and violate the property of perfectly matched sets $(A,B)$. This finishes the proof.
\end{proof}

The above proposition conclude that the reduction is a polynomial time parameter transformation. This finishes the proof of Theorem \ref{theorem-kernel}.
\section{Parameterized Algorithms}
\subsection{Parameterized by clique-width}\label{clique_width}

We refer to \cite{courcelle_engelfriet_2012, CourcelleLinear} for the details on $MSOL_1$ and linear-$EMSOL_{1}$. We recall that $MSOL_1$ is a type of $MSOL$ formula without quantifiers over edge sets, and linear-$EMSOL_{1}$ is an extension of $MSOL_1$ which allows for the search of an optimal solution with respect to some linear evaluation function.

We refer to \cite{CourcelleLinear,COURCELLE2000DAM} for the details on clique-width and $k$-expression. Courcelle, Makowsky and Rotics \cite{CourcelleLinear} showed that every decision problem expressible in $MSOL_{1}$  and every optimization problem expressible in linear-$EMSOL_{1}$ are fixed parameter tractable when parameterized by clique-width $k$ if the $k$-expression of the graph is given. Moreover, there exists an algorithm, which for a fixed $k$, and for the input $n$-vertex graph $G$, in time $O(n^9 \log n)$ outputs either a $(2^{3k+2}$ $-1)$-expression of $G$ or reports that $G$ has clique-width at least $k+1$ \cite{oumclique}. Combining the above known results, to show that $PMS$ is fixed parameter tractable when parameterized by clique-width, it will suffice to show that $PMS$ is expressible in linear-$EMSOL_{1}$.

 We construct the following linear-$EMSOL_1$ formula for PMS, which is similar to the $MSOL_1$ formula for MATCHING CUT given in \cite{Bonsma},

We first try to build an $MSOL_1$ formula $\psi(A,B)$ with free variables $A$ and $B$ over vertex sets to express that $(A,B)$ is a pair of perfectly matched sets, and then extend it to linear-$EMSOL_1$, which allows for the search of optimal $A$ and $B$ with respect to the cardinality of these sets.

\begin{align*}
    \psi(A,B)= (A\cap B = \emptyset)\wedge \forall v\in  A(matched(v,B)) \wedge \forall u\in B(matched(u,A))
\end{align*}
where we can write 
\begin{align*}
    (A\cap B = \emptyset) \Longleftrightarrow  \neg
   ( \exists v\in A(v\in B)).
\end{align*}
\begin{align*}
    matched(v,B) \Longleftrightarrow \exists u\in B(edge(u,v)) \wedge \neg (\exists u,x \in B (\neg(u=x)\wedge edge(u,v)\wedge edge(x,v))).
\end{align*}

An extension of $\psi(A,B)$ to linear-$MSOL_1$ is as follows,
\begin{align*}
    max\{|f(A)|: \langle G,f \rangle \  \models \psi(A,B)\}.
\end{align*}

where $f$ is an assignment of vertices to sets $A$ and $B$ such that $\psi(A,B)$ is satisfied. Note that we are only maximizing the size of $A$, this is sufficient for our purposes as the definition of PMS ensure that both $A$ and $B$ are of equal size.

\subsection{Parameterized by Distance to Cluster }\label{sec_cluster}
For a graph $G=(V,E)$, a vertex set $X\subseteq V$ is a cluster deletion set if $G-X$ is a cluster graph. The smallest size of a cluster deletion set of $G$ is called the
distance to cluster of $G$.
In this section, we consider PMS parameterized by distance to cluster of the input graph.
We note that for an $n$ vertex graph, a cluster deletion set of size at most $k$ can be computed in time $O(1.92^k\cdot n^2)$ \cite{DBLP:journals/mst/BoralCKP16}. Thus, we may assume that for the input graph, we have a cluster deletion set of size at most $k$.

In the remaining part of this section, we prove the following theorem.

\begin{theorem}\label{dctheorem}
There exists an algorithm that runs in time $O(3^{k}\cdot k^k\cdot poly(n))$ and solves PMS for the input $n$ vertex graph with the distance to cluster at most $k$.
\end{theorem}

For the input graph $G=(V,E)$, let $X \subseteq V$ be a cluster deletion set of size at most $k$. Further, let $\cal{C}=$ $\{C_1,C_2,..,C_l\}$ be the set of maximal cliques in $G-X$. 

For a vertex set $Y\subseteq V$, let $f: Y\to \{a,b,d\}$ be an assignment. We say that $f$ is \textit{valid} if every vertex in $f^{-1}(a)$ (resp. $f^{-1}(b)$) has at most one neighbor in $f^{-1}(b)$ (resp. $f^{-1}(a)$).
Given a valid assignment $f$, we say that a vertex $v \in f^{-1}(a)$ (resp. $v\in f^{-1}(b)$) is \textit{matched} in $f$ if $v$ has a neighbor $w$ in $f^{-1}(b)$ (resp. in $f^{-1}(a)$). In this case, we say $v$ is matched to $w$ in $f$, and both $v$ and $w$ are called \textit{matching neighbors} of each other.
We say that a non-empty set $W$ is matched in $f$ if every vertex in $W$ is matched. We say that a valid assignment $f: Y\to \{a,b,d\}$ is a \textit{feasible solution} if $Y=V(G)$ and $(f^{-1}(a),f^{-1}(b))$ is a pair of perfectly matched sets.

For an assignment $f$, we define $size(f)$ to be $|f^{-1}(a)|$ if $f$ is a feasible solution and $0$ otherwise. Our goal is to find an assignment $f$ with maximum $size(f)$; let $f_{opt}$ be such an assignment.
To find $size(f_{opt})$, we shall first guess its values on $X$ and then extend it to $V$. For this, we need the following observation.
\begin{observation}\label{obs:clique-types}
For a feasible solution $f$, every clique $C_i\in \cal C$ must be of one of the following types.
\begin{itemize}
    \item Type 1:  $V(C_i)\cap f^{-1}(a) \neq \emptyset$ and $V(C_i)\cap f^{-1}(b) = \emptyset$.
    \item Type 2:  $V(C_i)\cap f^{-1}(a) = \emptyset$ and $V(C_i)\cap f^{-1}(b) \neq \emptyset$.
    \item Type 3: There exist two distinct vertices $u,v\in V(C_i)$, such that $V(C_i)\cap f^{-1}(a) =\{u\}$ and $V(C_i)\cap f^{-1}(b) =\{v\}$. 
    \item Type 4: $V(C_i)\cap f^{-1}(a) = \emptyset$ and $V(C_i)\cap f^{-1}(b) = \emptyset$. 
\end{itemize}
\end{observation}

Let $f_x: X\to \{a,b,d\}$ be a valid assignment, an assignment $f: V(G)\to \{a,b,d\}$ is an \textit{extension} of $f_x$ if for every $ v\in X,\ f_x(v) = f(v)$. The intuition behind our algorithm is to consider  every valid $f_x$ and find its extension $f$ with maximum $size(f)$. We define $\footnotesize{\text{OPT}}(f_x)$ by
\begin{equation}
    \footnotesize{\text{OPT}}(f_x)= \max \{ size(f) \mid f \text{ is an extension of } f_x \}.
\end{equation}

We note that the $size(f_{opt})$ will be equal to the maximum $\footnotesize{\text{OPT}}(f_x)$ over all possible $f_x$.\\

Given $f_x: X\to \{a,b,d\}$, let $A_x = \{ v \mid  v\in X \land f_x(v)=a \land v \text{ is not matched in } f_x\}$ and $B_x = \{ v\mid v\in X \land f_x(v)=b \land v \text{ is not matched in }f_x\}$. Let $P_a=\{A_{x,1},\ldots,A_{x,r}\}$ be a partition of $A_x$ into non-empty sets ($P_a=\emptyset$ if $A_x$ is empty). Similarly, let $P_b=\{B_{x,1},\ldots,B_{x,s}\}$ be a partition of $B_x$ into non-empty sets ($P_b=\emptyset$ if $B_x$ is empty). 

We say that $f: V(G)\to \{a,b,d\}$ is a $(P_a,P_b)$-restricted extension of $f_x$ if the following hold.
\begin{itemize}
    \item $f$ is an extension of $f_x$.
    \item For every $A_{x,i}\in P_a$, exactly one of the following is true: no vertex of $A_{x,i}$ is matched in $f$ (OR) there exists a $p \in [l]$ such that $(A_{x,i},C_{p} \cap f^{-1}(b))$ is a pair of perfectly matched sets and $N(C_{p} \cap f^{-1}(b)) \cap f^{-1}(a)=A_{x,i}$.
    \item For every $B_{x,j}\in P_b$, exactly one of the following is true: no vertex of $B_{x,j}$ is matched in $f$ (OR) there exists a $q \in [l]$ such that $(C_q \cap f^{-1}(a),B_{x,j})$ is a pair of perfectly matched sets and $N(C_q \cap f^{-1}(a)) \cap f^{-1}(b)=B_{x,j}$.
\end{itemize}

We define $\footnotesize{\text{OPT}}(f_x,P_a,P_b)$ by
\begin{equation}
    \footnotesize{\text{OPT}}(f_x,P_a,P_b)=\max \{ size(f) \mid f \text{ is a } (P_a,P_b)\text{-restricted extensions of } f_x\}.
\end{equation}

 $\footnotesize{\text{OPT}}(f_x)$ will be equal to the maximum $\footnotesize{\text{OPT}}(f_x,P_a,P_b)$ over all possible partitions $(P_a,P_b)$. Formally, we make the following claim.
\begin{claim*}
 Given $f_x: X\to \{a,b,d\}$, it holds that:
\begin{equation}\label{eqn:optfx}
     \footnotesize{\text{OPT}}(f_x)=\max \{ \footnotesize{\text{OPT}}(f_x,P_a,P_b) \mid P_a,P_b \text{ are partitions of } A_x,B_x \text{ respectively } \}.
\end{equation} 
\end{claim*}
\begin{claimproof}
It is straightforward to see that L.H.S $\geq$ R.H.S in (\ref{eqn:optfx}). If $\footnotesize{\text{OPT}}(f_x)=0$, then L.H.S $\leq$ R.H.S in (\ref{eqn:optfx}). Else, if $\footnotesize{\text{OPT}}(f_x)\geq 1$, then let $\hat{f}$ be a extension of $f_x$ such that $size(\hat{f})= \footnotesize{\text{OPT}}(f_x)$,  In this case, $\hat{f}$ is a feasible solution. We will show that there exist partitions $P^*_{a},P^*_{b}$ of $A_x,B_x$ respectively such that $\hat{f}$ is a $(P^*_{a},P^*_{b})$-restricted extension of $f_x$. Recalling observation \ref{obs:clique-types}, let $T_2=\{C_{p1},C_{p2},...,C_{pr}\}$ be the set of all the cliques in $\cal C$ which are of type 2 in $\hat{f}$. Since $(\hat{f}^{-1}(a),\hat{f}^{-1}(b))$ is a pair perfectly matched sets, for every $C_{pi}\in T_2$, every vertex in $V(C_{pi})\cap \hat{f}^{-1}(b)$ has exactly one neighbor in $\hat{f}^{-1}(a)$ and that neighbor must belong to $A_x$.
In this case, if $T_2$ is empty then let $P^*_{a}=\emptyset$, otherwise let  $P^*_{a}=\{A_{x,1},A_{x,2},...,A_{x,r}\}$ such that $A_{x,i}= N(V(C_{pi})\cap \hat{f}^{-1}(b)) \cap \hat{f}^{-1}(a)$ for every $i\in [r]$. Since every vertex in $A_x$ has exactly one neighbor in $\hat{f}^{-1}(b)$ and that neighbor must belong to a clique in $T_2$, we have that $A_x= \bigcup_{A_{x,i}\in P^*_a} A_{x,i}$ and sets in $P^*_{a}$ are pairwise disjoint. Thus, we can conclude that $P^*_{a}$ is a partition of $A_x$. Further, let $T_1=\{C_{q1},C_{q2},...,C_{qs}\}$ be the set of all the cliques in $\cal C$ which are of type 1 in $\hat{f}$. If $T_1$ is empty then let $P^*_b =\emptyset$, otherwise let $P^*_{b}=\{B_{x,1},B_{x,2},...,B_{x,s}\}$ such that $B_{x,j}= N(V(C_{qj})\cap \hat{f}^{-1}(a)) \cap \hat{f}^{-1}(b)$ for every $j\in [s]$. Similar to $P^*_a$, arguments can be made to show that  $P^*_b$ is a partition of $B_x$. Thus, we have that for every $A_{x,i}\in P^*_a$, $(A_{x,i},C_{pi} \cap f^{-1}(b))$ is a pair of perfectly matched sets and $ N(V(C_{pi})\cap \hat{f}^{-1}(b)) \cap \hat{f}^{-1}(a)=A_{x,i}$, and for every $B_{x,j}\in P^*_b$, $(C_{qj} \cap f^{-1}(a),B_{x,j})$ is a pair of perfectly matched sets and $N(V(C_{qj})\cap \hat{f}^{-1}(a)) \cap \hat{f}^{-1}(b)= B_{x,j}$. We conclude that $\hat{f}$ is a $(P^*_{a},P^*_{b})$-restricted extension of $f_x$, and L.H.S $\leq$ R.H.S in (\ref{eqn:optfx}). This finishes the proof. 
\end{claimproof}

The number of valid assignments $f_x$ is at most $3^k$, and for a given valid assignment $f_x$, the number of partition pairs $(P_a,P_b)$ of $A_x,B_x$, is at most $k^k$.
Thus, to bound the running time of the algorithm, it suffices to find $\footnotesize{\text{OPT}}(f_x,P_a,P_b)$ in time $poly(n)$ for every fixed triple $(f_x,P_a,P_b)$, where $f_x$ is a valid assignment.
\par
Given a triple $(f_x,P_a,P_b)$, we say that a clique $C\in \cal C$ can perfectly match a set $A_{x,i} \in P_a$ if there exists $S\subseteq V(C)$ such that $(A_{x,i},S)$ is a pair of perfectly matched sets and $N(S)\cap f^{-1}_x(a) = A_{x,i}$. Similarly, we say that a clique $C\in \cal C$ can perfectly match a set $B_{x,j} \in P_b$ if there exists $T\subseteq V(C)$ such that $(T,B_{x,j})$ is a pair of perfectly matched sets and $N(T)\cap f^{-1}_x(b) = B_{x,j}$. We can verify if a clique $C \in \cal C$ can perfectly match $A_{x,i}$ in $poly(n)$ time by verifying if for every $v\in A_{x,i}$, there exists a vertex $u\in V(C)$ such that $N(u) \cap f^{-1}_x(a) = \{v\}$. We can similarly verify if $C$ can perfectly match $B_{x,j}$ as well. Further, we say that a $C$ can be of type 3 in an extension $f$ of $f_x$ if there exist $u,v \in V(C)$ such that $u\neq v$ and $N(u)\cap f^{-1}_x(b) = \emptyset$ and $N(v)\cap f^{-1}_x(a) = \emptyset$, this makes it possible for $u$ and $v$ to be the matching neighbors of each other in $f$ by setting $f(u)=a$, $f(v)=b$, and $f(w)=d$ for every $w\in V(C)\setminus \{u,v\}$.

\begin{figure}[h]

    \centering
    \begin{tikzpicture}


    \foreach \i in {1,2}{
     
            \node[shape=circle, draw, fill = black, scale= 0.3, font=\footnotesize,label={\scriptsize{$u_{\i}$}}] ({u\i}) at (-1+\i/2,-0.2){};
   }
   \foreach \i in {3,4,5}{
     
            \node[shape=circle, draw, fill = black, scale= 0.1, font=\footnotesize,] ({u\i}) at (-1+\i/2.1,-0.2){};
   }
   \foreach \i in {6}{
     
            \node[shape=circle, draw, fill = black, scale= 0.3, font=\footnotesize,label={\scriptsize{$u_{r}$}}] ({u\i}) at (-1.2+\i/2,-0.2){};
   }

     \foreach \i in {1,2}{
     
            \node[shape=circle, draw, fill = black, scale= 0.3, font=\footnotesize,label={\scriptsize{$v_{\i}$}}] ({v\i}) at (3+\i/2,-0.2){};
   }
   \foreach \i in {3,4,5}{
     
            \node[shape=circle, draw, fill = black, scale= 0.1, font=\footnotesize,] ({v\i}) at (3+\i/2.1,-0.2){};
   }
   \foreach \i in {6}{
     
            \node[shape=circle, draw, fill = black, scale= 0.3, font=\footnotesize,label={\scriptsize{$v_{s}$}}] ({v\i}) at (2.8+\i/2,-0.2){};
   }
   
    \foreach \i in {1,2}{
     
            \node[shape=circle, draw, fill = black, scale= 0.3, font=\footnotesize,label={\scriptsize{$y_{\i}$}}] ({y\i}) at (7+\i/2,-0.2){};
   }
   \foreach \i in {3,4,5}{
     
            \node[shape=circle, draw, fill = black, scale= 0.1, font=\footnotesize,] ({y\i}) at (7+\i/2.1,-0.2){};
   }
   \foreach \i in {6}{
     
            \node[shape=circle, draw, fill = black, scale= 0.3, font=\footnotesize,label={\scriptsize{$y_{l}$}}] ({y\i}) at (6.8+\i/2,-0.2){};
   }
   
    \foreach \i in {1}{
     
            \node[shape=circle, draw, fill = black, scale= 0.3, font=\footnotesize,label={-90:\scriptsize{$z_{\i}$}}] ({z\i}) at (2+\i/1.2,-3){};
   }
   \foreach \i in {2}{
     
            \node[shape=circle, draw, fill = black, scale= 0.3, font=\footnotesize,label={-90:\scriptsize{$z_{\i}$}}] ({z\i}) at (2+\i/1.2,-3){};
   }
   \foreach \i in {3,4,5}{
     
            \node[shape=circle, draw, fill = black, scale= 0.1, font=\footnotesize,] ({z\i}) at (3+\i/2.2,-3){};
   }
    \foreach \i in {6}{
     
            \node[shape=circle, draw, fill = black, scale= 0.3, font=\footnotesize,label={-90:\scriptsize{$z_{l}$}}] ({z\i}) at (1+\i/1.2,-3){};
   }


    \foreach \i in {1,3}{
            \draw[thin,black!60] (u\i)--(z1);

   }
    \foreach \i in {2,4}{
            \draw[thin,black!60] (u\i)--(z2);

   }
    \foreach \i in {4,5}{
            \draw[thin,black!60] (u\i)--(z3);

   }
    \foreach \i in {1,6}{
            \draw[thin,black!60] (u\i)--(z4);

   }
    \foreach \i in {2,5}{
            \draw[thin,black!60] (u\i)--(z5);
            \draw[thin,black!60] (u\i)--(z6);

   }
   
   \foreach \i in {1,3}{
            \draw[thin,black!60] (v\i)--(z1);

   }
    \foreach \i in {2,4}{
            \draw[thin,black!60] (v\i)--(z2);

   }
    \foreach \i in {4,5}{
            \draw[thin,black!60] (v\i)--(z3);

   }
    \foreach \i in {1,6}{
            \draw[thin,black!60] (v\i)--(z4);

   }
    \foreach \i in {2,5}{
            \draw[thin,black!60] (v\i)--(z5);
            \draw[thin,black!60] (v\i)--(z6);

   }
   \foreach \i in {1,2,4,5}{
            \draw[thin,black!20] (y\i)--(z\i);
            

   }

  

\end{tikzpicture}
\caption{An example of auxiliary graph $H$, vertex $u_i$ corresponds to $A_{x,i}\in P_a$, vertex $v_i$ corresponds to $B_{x,i}\in P_b$, vertices $y_j$ and $z_j$ corresponds to $C_j\in \cal{C}$. Darkened edges represent the edges which have $u_i$ or $v_j$ as an endpoint, and they have weight $n+1$ each. The other edges have weight $1$. }\label{fig:auxbi}
\end{figure}
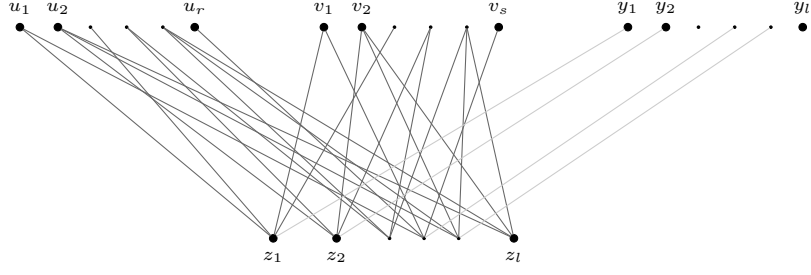

We now move on to find $\footnotesize{\text{OPT}}(f_x,P_a,P_b)$, for which we construct a weighted auxiliary bipartite graph $H=(L\cup R,E_H)$ such that
\begin{align*}
&L = \{u_i \mid A_{x,i}\in P_a\} \cup  \{v_i \mid B_{x,i}\in P_b\} \cup \{ y_j \mid C_j \in \cal{C}\},\\
&R= \{ z_j \mid C_j \in \cal{C} \}.
\end{align*}

and with the edge set 
\begin{align*}
    E_H = 
    &\{ u_iz_j\mid A_{x,i}\in P_a\land C_j\in {\cal{C}}\land C_j \text{ can perfectly match } A_{x,i}\} \\
    \cup
    &\{ v_iz_j\mid B_{x,i}\in P_b\land C_j\in {\cal{C}}\land C_j \text{ can perfectly match } B_{x,i}\} \\
    \cup 
    &\{ y_jz_j \mid  C_j\in {\cal{C}}\land C_j \text{ can be of type 3 }\}.
\end{align*}

The weight function $w: E_H \to \{n+1, 1\}$ is defined as: $w(e)=n+1$ if one of the endpoints of $e$ is $u_i$ or $v_j$ where $i\in[|P_a|]$ and $j\in[|P_b|]$, and $w(e)=1$ otherwise.

After constructing $H$, we find a maximum weight matching $M$ in $H$. 
Using $M$, we now extend $f_x$ to $f_m:V(G)\to \{a,b,d\}$ as follows.
For each $e\in M$:
\begin{enumerate}[ \text{Case} 1:]
    \item If $e=u_iz_j$, then let set $S\subseteq V(C_j)$ be such that $(A_{x,i},S)$ are perfectly matched sets and $N(S)\cap f^{-1}_x(a) = A_{x,i}$. If there are more than one such sets, we arbitrarily pick one; then $\forall u \in S$, we set $f_m(u)= b$, and $\forall u \in V(C_j)\setminus S$, we set $f_m(u)= d$.
   \item If $e=v_iz_j$, then let $S\subseteq V(C_j)$ be such that $(B_{x,i},S)$ are perfectly matched sets and $N(S)\cap f^{-1}_x(b) = B_{x,i}$. If there are more than one such sets, we arbitrarily pick one, then $\forall u \in S$, we set $f_m(u)= a$, and $\forall u \in V(C_j)\setminus S$, we set $f_m(u)= d$.
    \item If $e=y_jz_j$, then let $u,v\in V(C_j)$ be a pair of vertices such that $u\neq v$ and $N(u)\cap f^{-1}_x(b) = \emptyset$ and $N(v)\cap f^{-1}_x(a) = \emptyset$. If there are more than one such pair, we arbitrarily pick one of them. We then set $f_m(u)=a$ and $f_m(v)=b$ and $\forall w\in V(C_j)\setminus \{u,v\}$, $f_m(w)=d$.
\end{enumerate}
We assign $d$ to all the remaining vertices in $f_m$. \\

Observe that $f_m$ is a $(P_a, P_b)$-restricted extension of $f_x$. Further, we show that it is a valid assignment.
\begin{proposition}\label{fm_valid}
$f_m$ is a valid assignment.
\end{proposition}

\begin{proof}
Let $v$ be an arbitrary vertex in  $f_m^{-1}(a)$. We show that it has at most one neighbor in $f_m^{-1}(b)$. Similar arguments hold if $v \in f_m^{-1}(b)$.

\begin{enumerate}[ \text{Case} I:]
    \item $v \in X$ and $v$ is matched in $f_x$. In this case, $v$ has a neighbor in $f_x^{-1}(b) \subseteq f_m^{-1}(b)$ and by construction, no new neighbor of $v$ is set to $b$ in $f_m$.
    \item $v \in X$ and $v \in A_{x,i}$ for an $A_{x,i}\in P_a$. In this case, $v$ has no neighbor in $f_x^{-1}(b)$.  And if $u_iz_j \in M$ for a $j\in[l]$, then $v$ has exactly one neighbor in $f_m^{-1}(b)$ which by construction belongs to $C_j$.
    \item $v \in C_j$ for a $j\in[l]$. In this case, if $v$ is set to $a$ in $f_m$ due to the edge $y_jz_j\in M$, then $v$ has exactly one neighbor in $f_m^{-1}(b)$ which belongs to $C_j$ itself. Otherwise, if $v$ is set to $a$ in $f_m$ due to an edge $v_iz_j\in M$ for an $i\in[|P_b|]$, then $v$ has exactly one neighbor in $f_m^{-1}(b)$ which must belong to $B_{x,i}$.
    \end{enumerate}
\end{proof}

We now show that $f_m$ has maximum size among all $(P_a,P_b)$-restricted extensions of $f_x$.
\begin{proposition}\label{sizefm}
$size(f_m) = \footnotesize{\text{OPT}}(f_x,P_a,P_b)$.
\end{proposition}
\begin{proof}
Since $f_m$ is a $(P_a,P_b)$-restricted extension of $f_x$, and if $\footnotesize{\text{OPT}}(f_x,P_a,P_b)=0$, then $size(f_m) = \footnotesize{\text{OPT}}(f_x,P_a,P_b)$ due to the maximality of $\footnotesize{\text{OPT}}(f_x,P_a,P_b)$. Thus, we move on to assume that $\footnotesize{\text{OPT}}(f_x,P_a,P_b)\geq 1$. Let $\hat{f}$ be a $(P_a,P_b)$-restricted extension of $f_x$ such that $size(\hat{f})=\footnotesize{\text{OPT}}(f_x,P_a,P_b)$. In this case, $\hat{f}$ is a feasible solution. For every $A_{x,i}\in P_a$, let $p_i \in [l]$ be such that $(A_{x,i},C_{p_i}\cap \hat{f}^{-1}(b))$ is a pair of perfectly matched sets and $N(C_{p_i} \cap \hat{f}^{-1}(b)) \cap \hat{f}^{-1}(a)=A_{x,i}$. Similarly, for every $B_{x,j}\in P_b$, let $q_j \in [l]$ be such that $(C_{q_j}\cap \hat{f}^{-1}(a),B_{x,j})$ is a pair of perfectly matched sets and $N(C_{q_j} \cap \hat{f}^{-1}(b)) \cap \hat{f}^{-1}(a)=B_{x,j}$. This implies that for every $i\in[|P_a|]$ and  $j\in [|P_b|]$, $u_iz_{p_i}$ and $v_jz_{q_j}$ are edges in $H$. Also observe that every $C_{p_i}$ is of type 2, and every $C_{q_j}$ is of type 1 in $\hat{f}$. Further, let $T_3$ be the set of all cliques of type 3 in $\hat{f}$. To this end, let $\hat{M}=\{u_iz_{p_i}\mid i\in[|P_a|]\}\cup \{v_jz_{q_j}\mid j\in[|P_b|]\} \cup \{y_iz_{i}\mid C_i\in T_3\}$.  We have that for every distinct $i,j\in[|P_a|]$, $A_{x,i}\neq A_{x,j}$ and thus $C_{p_i}\neq C_{p_{j}}$. Similarly, for every distinct $i,j\in[|P_b|]$, $B_{x,i}\neq B_{x,j}$ and thus $C_{q_i}\neq C_{q_{j}}$. Thus, we observe that $\hat{M}$ is a matching in $H$; further $\hat{M}$ saturates every $u_i$ and every $v_j$, and the edges incident on these vertices have weight $n+1$ each; the other edges have weight $1$ each, and they are at most $l\leq n$. Thus, any maximum weight matching in $H$, must saturate all $u_i$ and $v_j$ vertices. Thus, $M$ (used for construction of $f_m$) also saturates every $u_i$ and every $v_j$. Case 1 and 2 of construction of $f_m$ ensures that every $A_{x,i}$ and every $B_{x,j}$ is matched in $f_m$, and Case 3 adds matching neighbors; this along with proposition \ref{fm_valid} ensures that
$(f_{m}^{-1}(a),f_{m}^{-1}(b))$ is a pair of perfectly matched sets and hence $f_m$ is a feasible solution.

 Since both $f_m$ and $\hat{f}$ are extensions of $f_x$, the number of vertices in $f_m^{-1}(a)\cap X$  and $\hat{f}^{-1}(a)\cap X$ are equal. Every vertex of $f_m^{-1}(a)\setminus X$ (resp.  $\hat{f}^{-1}(a) \setminus X$), must belong to a clique of type 1 or to a clique of type 3 in $f_m$ (resp.  $\hat{f}$), 
 Every vertex of $B_x$ must be matched to a vertex of a clique of type 1 in $f_m$ (resp. $\hat{f}$), and every vertex of a clique of type 1 which is in $f_m^{-1}(a)$ (resp. in $\hat{f}^{-1}(a)$) must be matched to a vertex of $B_x$, this concludes that the vertices of $f_m^{-1}(a)$ (resp. $\hat{f}^{-1}(a)$) which belong to a clique of type 1 in $f_m$ (resp. $\hat{f}$) are equal to $|B_{x}|$. Since every type 3 clique in $f_m$ (resp. $\hat{f}$) adds exactly one vertex in $f_m^{-1}(a)$ (resp. $\hat{f}^{-1}(a)$), if $|\hat{f}^{-1}(a)|>|f_m^{-1}(a)|$ then $\hat{f}$ must have more number of type 3 cliques than $f_m$. Further, since edges with weight $n+1$ are equal in both $\hat{M}$ and $M$, and each edge with weight 1 (edge $y_iz_i$ where $i\in[l]$) corresponds to a type 3 clique, thus $size(\hat{f})>size(f_m)$ implies that $w(\hat{M})>w(M)$, but since $M$ is a maximum weight matching in $H$, $w(\hat{M})\leq w(M)$, hence we have  $size(\hat{f})\leq size(f_m)$ (by contraposition). And due to the maximality of $\hat{f}$, we conclude that $size(\hat{f})= size(f_m)=\footnotesize{\text{OPT}}(f_x,P_a,P_b)$. This finishes the proof.

\end{proof}

For a given triple $(f_x,P_a,P_b)$, the construction of $f_m$ takes time $poly(n)$, thus we find $\footnotesize{\text{OPT}}(f_x,P_a,P_b)$ in time $poly(n)$. As discussed, this allows us to find $\footnotesize{\text{OPT}}(f_x)$ in time $O(k^{k}\cdot poly(n))$, which further allows us to find $size(f_{opt})$ in time $O(3^k\cdot k^{k}\cdot poly(n))$.

\subsection{Parameterized by Distance to Co-Cluster }\label{sec_co-cluster}

For a graph $G=(V,E)$, a vertex set $X\subseteq V$ is a co-cluster deletion set if $G-X$ is a co-cluster graph. The smallest size of a co-cluster deletion set of $G$ is called the
distance to co-cluster of $G$.
In this section, we consider PMS parameterized by distance to co-cluster of the input graph.
We note that for an $n$ vertex graph, a co-cluster deletion set of size at most $k$ can be computed in time $O(1.92^k\cdot n^2)$ \cite{DBLP:journals/mst/BoralCKP16}. Thus, we may assume that for the input graph, we have a co-cluster deletion set of size at most $k$.

In the remaining part of this section, we prove the following theorem.

\begin{theorem}\label{dcctheorem}
There exist an algorithm which runs in time $O(3^k \cdot poly(n))$ and solves PMS for the input $n$ vertex graph with the distance to co-cluster at most $k$.
\end{theorem}

For the input graph $G=(V,E)$, let $X \subseteq V$ be a co-cluster deletion set of size at most $k$. Further, let ${\cal {I}} = \{I_1,I_2,..,I_l\}$ be the set of maximal independent sets in $G-X$.

Recalling from Section \ref{sec_cluster}, for a vertex set $Y\subseteq V$, we say $f: Y\to \{a,b,d\}$ is an assignment. We refer Section \ref{sec_cluster} for the terminologies of \textit{valid assignment, extension of an assignment $f$, feasible solution,  $size(f)$} (size of an assignment $f$), and \textit{matched vertices} in an assignment $f$.
Similar to Section \ref{sec_cluster}, our goal is to find an assignment $f:V\to \{a,b,d\}$ such that $size(f)$ is maximized; let $f_{opt}$ be that assignment.

Let $f_x: X\to \{a,b,d\}$ be a valid assignment, we recall that an assignment $f: V(G)\to \{a,b,d\}$ is an \textit{extension} of $f_x$ if for every $ v\in X,\ f_x(v) = f(v)$. The intuition behind our algorithm is to consider every valid $f_x$ and find its extension $f$ with maximum $size(f)$. We define $\footnotesize{\text{OPT}}(f_x)$ by
\begin{equation}
    \footnotesize{\text{OPT}}(f_x)= \max \{ size(f) \mid f \text{ is an extension of } f_x \}.
\end{equation}

We note that the $size(f_{opt})$ will be equal to the maximum $\footnotesize{\text{OPT}}(f_x)$ over all possible $f_x$.
Since the number of assignments $f_x$ is at most $3^k$, to bound the running time of the algorithm, it suffices to find $\footnotesize{\text{OPT}}(f_x)$ in $poly(n)$.
 For a given valid assignment $f_x:X \to \{a,b,d\}$, let $A_x = \{ v\ \mid v\in X \land f_x(v)=a  \land v \text{ is not matched in } f_x\}$ and $B_x = \{ v\ \mid v\in X \land f_x(v)=b \land v \text{ is not matched in }f_x\}$. Let $\hat{f}$ be an extension of $f_x$ such that $size(\hat{f})=\footnotesize{\text{OPT}}(f_x)$. To find $\footnotesize{\text{OPT}}(f_x)$ consider the following cases.\\
\begin{enumerate}[ \text{Case} 1:]    \item \textbf{ Both $A_x$ and $B_x$ are non-empty.\\}
    If $size(\hat{f})\geq 1$, then $\hat{f}$ is a feasible solution, and all the vertices of $A_x$ and $B_x$ must be matched to vertices of $V\setminus X$ in $\hat{f}$. Further, we make the following claim.
   
      \begin{claim*}
        If $size(\hat{f})\geq 1$, then all the vertices of $V\setminus X$ which are set to either $a$ or $b$ in $\hat{f}$ belong to a single maximal independent set $I_i\in \cal I$.
      \end{claim*}
      \begin{claimproof}
      Since both $A_x$ and $B_x$ are non-empty and they must be matched to vertices of $V\setminus X$ in $\hat{f}$, there must be two vertices $u,v\in V\setminus X$ such that $\hat{f}(u)=a$ and $\hat{f}(v)=b$, and $u,v$ matched to a vertex of $B_x$ and $A_x$ respectively. In this case, $u$ and $v$ must belong to the same $I_i\in \cal I$, otherwise they will be neighbors in opposite sets and $\hat{f}$ will not be a feasible solution. Similarly, any other vertex $w\in V\setminus X$ if assigned a value $a$ or $b$ in $\hat{f}$ should be from $I_i$, otherwise it will be a neighbor of either $u$ or $v$ in an opposite set.
       \end{claimproof}

Using the above claim, we construct an extension $f$ of $f_x$ as follows. We guess a set $I_i\in \cal I$ who's vertices may be set to $a$ or $b$ in $f$, and set vertices of all other $I_j\in \cal I$ to $d$. 
Since $I_i$ is an independent set, a vertex $u\in I_i$ if set to $a$ or $b$ must be matched to a vertex of $B_x$ or $A_x$ respectively. We say that a vertex $u\in I_i$ can perfectly match a vertex $v\in A_x$ (resp. $v\in B_x$) if $N(u)\cap f_x^{-1}(a)=\{v\}$ (resp. $N(u)\cap f_x^{-1}(b)=\{v\}$). Consider an auxiliary bipartite graph $H$ with vertex sets $I_i\cup A_x\cup B_x$, we connect $u\in I_i$ to $v\in (A_x\cup B_x)$ in $H$ if $u$ can perfectly match $v$. We find a maximum matching $M$ in $H$. For each edge $e\in M$, let $u$ be the endpoint which belong to $I_i$ and $v$ belong to $A_x\cup B_x$, we set $f(u)=a$ if $f(v)=b$ else we set $f(u)=b$. Once we process all the edges of $M$, we set $f(x)=d$ for every remaining unassigned vertex $x\in I_i$. We can observe that if $\hat{f}$ is a feasible solution and we correctly guessed $I_i\in \cal {I}$, then $M$ saturates $A_x\cup B_x$ and $f$ is a feasible solution with $size(f)=size(\hat{f})$. Construction of $f$ takes time $poly(n)$, and there are at most $n$ distinct sets to guess from $\cal{I}$. Further, if for no $I_i\in \cal{I}$ the constructed $f$ is a feasible solution, then we conclude that $size(\hat{f})= \footnotesize{\text{OPT}}(f_x)=0$.
\\

\item \textbf{ $A_x$ or $B_x$ is empty.\\}
We assume that $B_x$ is empty and $A_x$ may be non-empty, the case when $A_x$ is empty and $B_x$ may be non-empty will be similar. If $size(\hat{f})\geq 1$, then $\hat{f}$ is a feasible solution, and if $A_x$ is non-empty then all the vertices of $A_x$ must be matched to vertices of $V\setminus X$ in $\hat{f}$. Further, we make the following claim.
  
      \begin{claim*}
        If $size(\hat{f})\geq 1$, then the following holds.
        \begin{itemize}
            \item If $A_x$ is non-empty, then at most one vertex from $V\setminus X$ is assigned $a$ in $ \hat{f}$.
            \item If $A_x$ is empty, then at most two vertices from $ V\setminus X$ are assigned $a$ in $ \hat{f}$.
        \end{itemize}
      \end{claim*} 
      \begin{claimproof}
    Let there be two distinct vertices $u_1,u_2\in (V\setminus X)$ such that $\hat{f}(u_1)$=$\hat{f}(u_2)$=$a$. Since $B_x$ is empty and $\hat{f}$ is a feasible solution, $u_1,u_2$ must be matched to vertices of $V\setminus X$, in this case let $u_1,u_2$ be matched to $v_1,v_2\in (V\setminus X) \cap \hat{f}^{-1}(b)$ respectively. Observe that $u_1$ and $v_2$ should belong to the same $I_i\in \cal I$ and $u_2$ and $v_1$ should belong to the same $I_j\in \cal I$ where $I_i\neq I_j$, otherwise they will be neighbors of each other. This implies that both $\hat{f}^{-1}(a)$ and $\hat{f}^{-1}(b)$ contain vertices from two distinct $I_i$ and $I_j$. First, if $A_x$ is non-empty, then there must be a vertex $x\in V\setminus X$ such that $\hat{f}(x)=b$ and $x$ is matched to a vertex in $A_x$ (as $\hat{f}$ is a feasible solution), since $u_1$, $u_2$ belong to distinct maximal independent sets, at least one of them is a neighbor of $x$ in opposite set, this contradicts that both $u_1$ and $u_2$ belong to $\hat{f}^{-1}(a)$.
    Second, if $A_x$ is empty, then assigning $a$ to any other vertex $z\in V\setminus X$ will make it a neighbor of either $v_1$ or $v_2$ (which are already matched in $\hat{f})$, this finishes the proof.
      \end{claimproof}

Using above claim, we construct an extensions $f$ of $f_x$ as follows. We guess a set $S\subseteq V\setminus X$ of size at most one (if $A_x$ is non-empty) or at most two (if $A_x$ is empty). Note that $S$ can be empty and it is our guess on the vertices of $V\setminus X$ which are assigned $a$ in $\hat{f}$. We set $f(u)=a$ for every $u\in S$. Since $B_x$ is empty, every vertex in $f_x^{-1}(b)$ is already matched in $f_x$,
and it is necessary that every vertex in $S$ has no neighbor in $f_x^{-1}(b)$. If a vertex in $S$ has a neighbor in $f_x^{-1}(b)$, then we guess another $S$. Now, for every $u\in A_x\cup S$ we check if there exists a vertex $v \in V\setminus (X\cup S)$ such that $N(v)\cap (f_x^{-1}(a)\cup S) = \{u\}$, if there are more than 1 such vertices, we arbitrarily pick one such vertex $v$ and set $f(v)=b$. In the end, we assign $f(x)=d$ to every remaining unassigned vertex  $x\in V\setminus (X\cup S)$. We can observe that if $\hat{f}$ is a feasible solution and we correctly guessed $S$, then $f$ is a feasible solution and $size(f)=size(\hat{f})$. Construction of $f$ takes time $poly(n)$, and there are at most $O(n^2)$ distinct sets $S$ to guess from $V\setminus X$. The maximum size of a constructed $f$ among all the guesses of $S$ will be equal to $size(\hat{f})=\footnotesize{\text{OPT}}(f_x)$.
Further, if for no $S$ the constructed $f$ is a feasible solution, then we conclude that $size(\hat{f})=\footnotesize{\text{OPT}}(f_x)=0$.
\\

\end{enumerate}

As discussed, there are at most $3^k$ assignments $f_x$, and finding $\footnotesize{\text{OPT}}(f_x)$ for each $f_x$ takes time $poly(n)$. Thus, it takes time $O(3^k\cdot poly(n))$ to find $size(f_{opt})$. This finishes the proof.
\subsection{Parameterized by Treewidth}\label{section-tw}

We refer to \cite{DBLP:books/sp/CyganFKLMPPS15,planartreewidth} for the details on tree decomposition and treewidth. We recall here basic definitions of tree decomposition and treewidth from \cite{DBLP:books/sp/CyganFKLMPPS15}. A \textit{tree decomposition} of a graph $G$ is a pair $\mathcal{T}= (T,\{\beta_t\}_{t\in V(T)})$, where $T$ is a tree and every node in $V(T)$ assigned a vertex subset $\beta_t \subseteq V(G)$, also called the bag of $t$ such that the following holds. (i) For every $v\in V(G)$, there exists a node $t\in V(T)$ such that $v\in \beta_t$;
(ii) for every edge $e\in E(G)$, there exists a node $t\in V(T)$ such that $V(e)\subseteq \beta_t$;
(iii) For every $v\in V(G)$, let $T_v = \{t\mid t\in V(T)\wedge v\in \beta_t\}$, that is $T_v$ is the set of all the nodes of $T$ that contain $v$ in their bags, then $T[T_v]$ induces a connected subgraph of $T$. Further, a tree decomposition $\mathcal{T}= (T,\{\beta_t\}_{t\in V(T)})$ is said to be a rooted tree decomposition if $T$ is rooted at some node $r\in V(T)$.

The width of the tree decomposition $\mathcal{T}$ is $\max\{ |\beta_t|-1 \ |\ t\in V(T)\}$. The treewidth of a graph $G$ denoted by $tw(G)$ is minimum width over all the possible tree decompositions of $G$.

A useful property of tree decomposition is the existence of \textit{nice tree decomposition}, we refer to \cite{DBLP:books/sp/Kloks94,DBLP:books/sp/CyganFKLMPPS15} for the details on nice tree decomposition, and recall a definition of nice tree decomposition here.

A rooted tree decomposition $\mathcal{T}= (T,\{\beta_t\}_{t\in V(T)})$ is nice if $T$ is rooted at a node $r$ with empty bag, i.e. $\beta_r = \emptyset$, and every node of $T$ belongs to one of the following type.
\begin{itemize}
    \item \textit{Leaf node}: a leaf $l$ of $T$ and $\beta_t = \emptyset$.
    \item \textit{Introduce node}: a node $t$ of $T$ with exactly one child $c$ such that $\beta_c\subseteq \beta_t $ and  $\beta_t \setminus \beta_c = \{v\} $ for some vertex $v\in V(G)$, vertex $v$ is said to be introduced at $t$.
    \item \textit{Forget node}: a node $t$ of $T$ with exactly one child $c$ such that $\beta_t\subseteq \beta_c $ and  $\beta_c \setminus \beta_t = \{v\} $ for some vertex $v\in V(G)$, vertex $v$ is said to be forgot at $t$.
    \item \textit{Join node}: a node $t$ of $T$ with exactly two children $c_1$ and $c_2$ such that $\beta_t= \beta_{c_1} = \beta_{c_2}$.
\end{itemize}

It is also known that a tree decomposition with width at most $tw$ can be converted into a nice tree decomposition with width at most $tw$ and number of nodes at most $O(tw\cdot n)$ in time $poly(n)$ \cite{DBLP:books/sp/CyganFKLMPPS15}.

Further, in a nice tree decomposition $\mathcal{T}= (T,\{\beta_t\}_{t\in V(T)})$, for a $t\in T$, we denote $\gamma_t$ to be the union of bags of all the nodes which belong to the subtree rooted at $t$, and $G_t= G[\gamma_t]$.

Now that we discussed preliminaries for this section, in the remaining part of this section, we prove the following theorem.
\begin{theorem}\label{twtheorem}
Given a nice tree decomposition of width at most $tw$ for input $n$-vertex graph $G$, there exists an algorithm that runs in time $O(10^{tw}\cdot poly(n))$ and solves PMS for $G$.  The algorithm can work even if such a tree decomposition is not given, at an expense of higher constant at the base of the exponent.
\end{theorem}

As expected, we will perform a bottom up dynamic programming on a rooted nice tree decomposition $(T,\beta)$ of the input graph $G$. 
To this end, we first define the structure of the sub-problem and corresponding memory table entry.

At each node $t\in V(T)$, we denote an entry of our table as $m_t(A_t,B_t, p, n_A, n_B)\in \{0,1\}$, where $A_t\subseteq \beta_t$, $B_t\subseteq \beta_t$, $A_t\cap B_t = \emptyset$, $p:(A_t\cup B_t)\to \{0,1\}$, $n_A,n_B\in \{0,1,...,n\}$. Precisely, $m_t(A_t,B_t, p, n_A, n_B)=1$ indicates if there exist two disjoint sets $A,B\subseteq V(G_t)$ such that the following holds,

\begin{itemize}
    \item $A\cap \beta_t= A_t$ and $B\cap \beta_t= B_t$,
    \item $A\setminus \beta_t$ has exactly one neighbor in $B$ and $B\setminus \beta_t$ has exactly one neighbor in $A$,
    \item every $v\in A_t$ has exactly $p(v)$ neighbors in $B\setminus \beta_t$ and similarly every $v\in B_t$ has exactly $p(v)$ neighbors in $A\setminus \beta_t$,
    \item $|A|=n_A$ and $|B|=n_B$.
\end{itemize}

Observe that $G$ has PMS of size $k$ if and only if $m_r(\emptyset,\emptyset,p,k,k)$ is $1$, where $r$ is the root node of $T$ and $p$ is empty.
Observe that the number of entries at each node are bounded by $\sum_{i\in[tw+1]}{tw+1 \choose i}\cdot 2^i\cdot 2^i \cdot poly(n)$, which is equal to $(1+4)^{tw+1}\cdot poly(n)$. Thus, there are at most $O(5^{tw}\cdot poly(n))$ entries $m_t(.)$ for each node $t$.

We now give recursive formulas for every type of node while assuming that the entries for its children have already been computed.\\

\begin{itemize}

    \item \textbf{Leaf node $t$} : $\beta_t= \emptyset$. We set $m_t(\emptyset,\emptyset,p \text{ = empty},0,0) = 1$, and set all other entries to $0$.
\\
\item\textbf{Introduce node}: vertex $v$ introduced in $\beta_t$ and $c$ is the child of $t$ in $T$.
\\
For the calculation of an entry $m_t(A_t,B_t,p,n_A,n_B)$, let $p_c: ((A_t\cup B_t)\setminus \{v\}) \to \{0,1\}$ such that $p_c(x)=p(x)$ for every vertex $x\in ((A_t\cup B_t)\setminus \{v\})$.\\

        \[m_t(A_t,B_t,p,n_A,n_B)=
        \begin{cases}
            m_c(A_t\setminus \{v\},B_t,p_c,n_A-1,n_B), & if\ v\in A_t\wedge p(v)=0, \\
            0,               & if\ v\in A_t\wedge p(v)=1.\\
            m_c(A_t,B_t\setminus \{v\},p_c,n_A,n_B-1), & if\ v\in B_t\wedge p(v)=0, \\
            0,               & if\ v\in B_t\wedge  p(v)=1.\\
            m_c(A_t,B_t,p_c,n_A,n_B), & if \ v  \not \in B_t \wedge v  \not \in A_t.
        \end{cases}\]
        \\
  
  To observe the correctness of the above formula, first consider the case when $v\in A_t$, if $p(v)=1$, then $v$ is bound to have a matching neighbor in $B\setminus \beta(t)$. By the definition of tree decomposition, $v$ has no neighbor in $V(G_t)\setminus \beta(t)$, thus we set this entry to $0$, and if $p(v)=0$ then the value $m_c(A_t\setminus \{v\},B_t,p_c,n_A-1,n_B)$ must also hold for $m_t(A_t,B_t,p,n_A,n_B)$ as the size of $A_t\setminus \{v\}$ is $n_A-1$. The case when $v\in B_t$ is similar to the above. Further, calculation of every entry as per above formula takes $O(1)$ time to compute. 
\\
\item\textbf{Forget node}: vertex $v$ forgot from $\beta_t$ and $c$ is child of $t$ in $T$.\\
For the calculation of $m_t(A_t,B_t,p,n_A,n_B)$ at forget node, we search if of two disjoint sets $A,B\subseteq V(G_t)$ exist which satisfy $m_t(A_t,B_t,p,n_A,n_B)=1$. For this purpose, we need to consider all the cases where $v$ may belong to $A$, $B$ or to none. We then choose the best possible outcome.

\begin{enumerate}[\text{Case} 1:]
\item  $v$ belongs to $A$. In this case $v$ must have exactly one neighbor in $B$, let $N_{B_t} = N(v)\cap B_t$.
    Let $p_0:(A_t\cup B_t\cup \{v\}) \to \{0,1\}$ such that $p_0(x)= p(x)$ for every $x\in (A_t\cup B_t\cup \{v\})\setminus (N_{B_t}\cup \{v\})$ and $p_0(u)= 0$ for every $u \in (N_{B_t}\cup \{v\})$. Further, let $p_1:(A_t\cup B_t\cup \{v\}) \to \{0,1\}$ such that $p_1(x)= p(x)$ for every $x\in (A_t\cup B_t\cup \{v\})\setminus \{v\}$ and $p_1(v)= 1$.
    \[ M_A=
    \begin{cases}
            m_c(A_t\cup \{v\},B_t,p_0,n_A,n_B), & \text{if } |N_{B_t}|=1,\\
            m_c(A_t\cup \{v\},B_t,p_1,n_A,n_B), & \text{if } |N_{B_t}|=0,\\
            0, & \text{if } |N_{B_t}|\geq 2.
    \end{cases}\]
    
    $M_A$ captures the best possible value when $v$ belongs to $A$, and for the correctness of the above formula, observe that if $v$ belongs to $A$, then it must have exactly one neighbor $u$ in $B$. In such a case, if $u$ belongs to $B_t$, then $v$ cannot have any other neighbor in $B\setminus B_t$ and $u$ cannot have any other neighbor in $A\setminus (A_t\cup \{v\})$, hence $p_0(v)$ and $p_0(u)$ should be $0$. Else if $u$ belongs to $B\setminus B_t$ then $v$ cannot have a neighbor in $B_t$ and $p_1(v)$ should be $1$, finally if $v$ has two or more neighbors in $B_t$ then no $A$ containing $v$ can satisfy required properties.
    \\
    \item $v$ belongs to $B$.
In this case $v$ must have exactly one neighbor in $A$, let $N_{A_t} = N(v)\cap A_t$.
Let $p_0:(A_t\cup B_t\cup \{v\}) \to \{0,1\}$ such that $p_0(x)= p(x)$ for every $x\in (A_t\cup B_t\cup \{v\})\setminus (N_{A_t}\cup \{v\})$ and $p_0(u)= 0$ for every $u \in (N_{A_t}\cup \{v\})$. Further, let $p_1:(A_t\cup B_t\cup \{v\}) \to \{0,1\}$ such that $p_1(x)= p(x)$ for every $x\in (A_t\cup B_t\cup \{v\})\setminus \{v\}$ and $p_1(v)= 1$.
\[M_B=
    \begin{cases}
            m_c(A_t,B_t\cup \{v\},p_0,n_A,n_B), & \text{if } |N_{A_t}|=1,\\
            m_c(A_t,B_t\cup \{v\},p_1,n_A,n_B), & \text{if } |N_{A_t}|=0,\\
            0, & \text{if } |N_{A_t}|\geq 2,\\
    \end{cases}
    \]
    
    $M_B$ captures the best possible value when $v$ belongs to $B$ and correctness follows from similar arguments to the previous case.
    \\
    \item  $v$ is deleted.
    \begin{align*}
        M_{\emptyset}=m_c(A_t,B_t,p,n_A,n_B).
    \end{align*}
\end{enumerate}

The recursive formula for forget node is as follows.
\begin{align*}
    m_t(A_t,B_t,p,n_A,n_B)= max\{M_A,M_B, M_{\emptyset}\}.
\end{align*}

We observe that the calculation of each entry as per the above formula takes time $poly(n)$.
\\
\item\textbf{Join node}: $c_1$ and $c_2$ are children of $t$.

Let a pair $(p_{c_1},p_{c_2})$ be such that $p(v)= p_{c_1}(v)+p_{c_2}(v)$ for every $v\in (A_t\cup B_t)$. Let $\cal F$ be the family of all possible pairs for a given $p$, observe that the size of $\cal F$ is bounded by $2^{tw}$.

\begin{multline*}
    m_t(A_t,B_t,p,n_A,n_B)= 
    \max \{m_{c_1}(A_t,B_t,p_{c_1},|A_t|+i,|B_t|+j)\cdot m_{c_2}(A_t,B_t,p_{c_2},n_A-i,n_B-j)\ |\ \\
    (p_{c_1},p_{c_2})\in \calf \land 0\leq i \leq (n_A-|A_t|) \land  0\leq j \leq (n_B-|B_t|) \} .
\end{multline*}

For the correctness, observe that $p(v)= p_{c_1}(v)+p_{c_2}(v)$ enforces every vertex $v$ in $A_t$ (resp. $B_t$) has at most $p(v)$ neighbor in $B\setminus \beta(t)$ (resp $B\setminus \beta(t)$), only one of the $G_{c_1}$ or $G_{c_2}$ can have that neighbor (if $p(v)=1$ otherwise none). Similarly, $i$ and $j$ denote the distribution of $A\setminus A_t$ and $B\setminus B_t$ vertices in the sub graphs $G_{c_1}$ and $G_{c_2}$ to ensure that the number of distinct vertices in $A$ and $B$ remain $n_A$ and $n_B$ respectively, and we don't double count any  of the vertex.

Calculation of a single entry at a join node takes $2^{tw}\cdot poly(n)$ time. There can be at most  $5^{tw}\cdot poly(n)$ entries at each node, thus at a join node calculation of all the entries takes time $10^{tw}\cdot poly(n)$, and  this is the bottleneck for the running time of the algorithm.

\end{itemize}
\section{Exact Algorithm}\label{sec_exact}

In this section, we prove the following result.
\begin{theorem}\label{thm-exact-PMS}
There is an algorithm that accepts a graph $G$ on $n$ vertices and finds a largest pair of perfectly matched sets of $G$ in time $O^*(1.966^n)$ time.
\end{theorem}

We shall prove Theorem \ref{thm-exact-PMS} by making use of the algorithms in the following two lemmas.
\begin{lemma}\label{PMS-size-lemma-one}
There's an $O^*\left(\dbinom{n}{k}\right)$ algorithm to test if $G$ has a pair of perfectly matched sets of size $k$, and find one such pair if it exists.
\end{lemma}
\begin{proof}
The algorithm is the following: for each subset $A$ of size $k$, first find $C(A)$ which we define to be the set of vertices in $V \setminus A$ with exactly one neighbor in $A$.
If every $v \in A$ has at least one neighbor in $C(A)$, then mark one such neighbor for each $v \in A$ as $f(v)$ and let $B=\{f(v)|v \in A\}$. Then $(A,B)$ forms a pair of perfectly matched sets. If for no $A$ can we find a corresponding $B$, then $G$ has no pair of perfectly matched sets of size $k$.
\end{proof}

\begin{lemma}\label{PMS-size-lemma-two}
If there's an $O^*(\alpha^n)$ algorithm to solve PERFECT MATCHING CUT on graphs of size $n$, then there's an $O^*\left(\dbinom{n}{2k}{\alpha}^{2k}\right)$ algorithm to test if $G$ has a pair of perfectly matched sets of size $k$, and find one if it exists. In particular, by the result of \cite{LT2020}, there's an $O^*\left(\dbinom{n}{2k}{1.2721}^{2k}\right)$ algorithm to do so.
\end{lemma}

\begin{proof}
The algorithm consists of picking every subset $S$ of size $2k$ and checking if $G[S]$ has a perfect matching cut.
\end{proof}

\begin{proof}{\bf of Theorem \ref{thm-exact-PMS}}
The main idea is to play off the bounds of Lemma \ref{PMS-size-lemma-one} and Lemma \ref{PMS-size-lemma-two}. When $k$ is close to $n/2$, the algorithm of Lemma \ref{PMS-size-lemma-two} is faster while for slightly smaller $k$, the algorithm of Lemma \ref{PMS-size-lemma-one} is faster.

Let $\varepsilon <1/2$, which we will fix shortly.
For $1 \leq k \leq \varepsilon n$, we run the algorithm of Lemma \ref{PMS-size-lemma-one} and for $\varepsilon n <k \leq n$, we run the algorithm of Lemma \ref{PMS-size-lemma-two}.
Thus we can find a largest pair of perfectly matched sets in time $O^*(T(n))$, for the following value of $T(n)$.
\begin{align*}
T(n)&=\sum_{k=0}^{\varepsilon n} \dbinom{n}{k} + \sum_{k=\varepsilon n}^{n/2} \dbinom{n}{2k}{\alpha}^{2k} \\
&\leq 2^{nH(\varepsilon)} + n\dbinom{n}{2 \varepsilon n}{\alpha}^{2\varepsilon n} \\
& \leq 2^{nH(\varepsilon)}+2^{nH(1-2\varepsilon)}{\alpha}^{2\varepsilon n}.
\end{align*}

The second line follows from the fact that the terms of the second sum are decreasing with $k$ for 
$\varepsilon \geq \dfrac{1}{2+\alpha}$, and we also use the upper bound of $2^{H(tn)}$ on $\sum_{k=1}^{tn}\dbinom{n}{k}$ for $t \leq 1/2$.

The optimal value of $\varepsilon$ that minimizes the RHS is found by equating the two terms; i.e. by solving the equation $H(\varepsilon)=H(1-2\varepsilon)+ (2\log_2 \alpha) \varepsilon$.
For $\alpha=1.2721$, this yields $\varepsilon \sim 0.4072$ and $2^{H(\varepsilon)} \sim 1.96565$.
Thus, we obtain $T(n)=O(1.966^n)$.
\end{proof}

\section{PMS for Planar Graphs }\label{section-planar-hardness}

In this section, we prove the following theorem.
\begin{theorem}\label{planarNP_hardness}
PMS is NP-hard for planar graphs.
\end{theorem}

It is known that INDEPENDENT SET is NP-hard on planar graphs. We give a polynomial time reduction from INDEPENDENT SET to PMS.

Let $(G,k)$ be an instance of INDEPENDENT SET, where $G$ is a planar graph and $k\in \mathbb{N}$, and we need to decide if $G$ contains an independent set of size $k$. Let $V=\{v_1,v_2,...,v_n\}$ be vertices of $G$, we construct a graph $G'$ as follows.
\begin{figure}[H]
    \centering
    
    \begin{tikzpicture}[]
\node [shape=circle, draw, fill=black,scale =0.6, label={120:\scriptsize{$v_i$}}] (v) at (0,1) {};
\node [shape=circle, draw, fill=black,scale =0.6] (u) at (-1,-0.5) {};
\node [shape=circle, draw, fill=black,scale =0.6,label={180:\scriptsize{$v_j$}}] (x) at (1,-0.5) {};
\node [shape=circle, draw, fill=black,scale =0.6] (y) at (0,2.5) {};
\draw (v)--(u) (v)--(x) (v)--(y);

\node [shape=circle, fill=blue!50,scale =0.3, label={270:\scriptsize{\textcolor{blue!65}{$v''_i$}}}] (vi') at (5,0.5) {};
\node [shape=circle, draw, fill=black,scale =0.6, label={120:\scriptsize{$v'_i$}}] (vi) at (5,1) {};
\node [shape=circle, draw, fill=black,scale =0.6] (u) at (4,-0.5) {};
\node [shape=circle, draw, fill=black,scale =0.6,label={180:\scriptsize{$v'_j$}}] (vj) at (6,-0.5) {};
\node [shape=circle, fill=blue!50,scale =0.3, label={270:\scriptsize{\textcolor{blue!65}{$v''_j$}}}] (vj') at (6,-1) {};

\node [shape=circle, draw, fill=black,scale =0.6] (y) at (5,2.5) {};
\draw (vi)--(u) (vi)--(vj) (vi)--(y);

\node [shape=circle, fill=blue!50,scale =0.4, label={180:\scriptsize{\textcolor{blue!65}{$e_{ij}^1$}}}] (e1) at (6.5,0.5) {};

\node [shape=circle, fill=blue!50,scale =0.4, label={120:\scriptsize{\textcolor{blue!65}{$e_{ij}^2$}}}] (e2) at (7.25,0.5) {};
\node [shape=circle, fill=blue!50,scale =0.4, label={120:\scriptsize{\textcolor{blue!65}{$e_{ij}^3$}}}] (e3) at (8,0.5) {};
\node [shape=circle, fill=blue!50,scale =0.4, label={20:\scriptsize{\textcolor{blue!65}{$e_{ij}^4$}}}] (e4) at (8.75,0.5) {};

\draw[thin,blue!40]  (e1) -- (e2)--(e3)--(e4);

\draw [->,>=stealth] (2,0.5) -- (3,0.5);

\draw[thin,blue!40,bend left] (vi) to ({e1});
\draw[thin,blue!40,bend left] (vi) to ({e4});
\draw[thin,blue!40,bend right] (vj) to ({e2});
\draw[thin,blue!40,bend right] (vj) to ({e3});

\draw[thin,blue!40] (vi) to ({vi'});
\draw[thin,blue!40] (vj) to (vj');

\end{tikzpicture}

    \caption{Reduction from INDEPENDENT SET to PMS where input and constructed graphs are planar.   }
    \label{fig:planar_reduction1}
\end{figure}
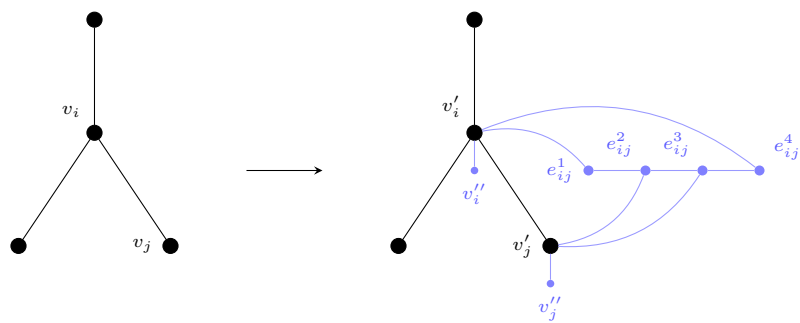

\begin{itemize}
    \item Create a vertex set $V'=\{v'_i\mid v_i\in V\}$, and connect every $v'_i$ to $v'_j$ if $v_iv_j\in E(G)$.
    \item Create a vertex set $V''=\{v''_i\mid v_i\in V\}$, and connect $v'_i$ to $v''_i$ for every $i\in[n]$.
     \item Create a vertex set $X_{e}= \{e_{ij}^1,e_{ij}^2,e_{ij}^3,e_{ij}^4\mid v_iv_j\in E(G) \land (i<j)\}$. Further, for every edge $v_iv_j\in E(G)$ where $i<j$, connect $e_{ij}^1$ to $e_{ij}^2$, $e_{ij}^2$ to $e_{ij}^3$, and $e_{ij}^3$ to $e_{ij}^4$, i.e. create a path on 4 vertices. Further, connect $v'_{i}$ to $e_{ij}^1$ and $e_{ij}^4$, and connect $v'_j$ to $e_{ij}^2$ and $e_{ij}^3$.
\end{itemize}

Observe that none of the introduced edges in $G'$ is a crossing edge if $G$ was embedded in a plane without crossing edges, and hence $G'$ is a planar graph. Further, the above construction takes time polynomial in size of $G$. Following proposition proves the correctness of the reduction.

\begin{proposition}
 $G$ has an independent set of size $k$ if and only if $G'$ has a pair of perfectly matched sets of size $k+2\cdot |E(G)|$.
\end{proposition}
\begin{proof}
For the forward direction, let $I$ be an independent set of size $k$ in $G$. We construct $(A,B)$ as follows. 
\begin{enumerate}
    \item For every $v_i\in I$, put $v'_i$ in $A$ and $v''_i$ in $B$.
    \item For every edge $v_iv_j \in E(G)$ such that $i<j$, if $v_i$ is in $I$ (then certainly $v_j$ is not in $I$), then add $e^1_{ij}, e^4_{ij}$ to $A$ and $e^2_{ij}, e^3_{ij}$ to $B$, else if $v_i$ is not in $I$, then add $e^1_{ij}, e^4_{ij}$ to $B$ and $e^2_{ij}, e^3_{ij}$ to $A$.
\end{enumerate}
 We claim that $(A,B)$ is a pair of perfectly matched sets and $|E(A,B)|= k+2\cdot |E(G)|$. For the proof of the claim, consider the following arguments. 
 \begin{itemize}
     \item Consider the vertices of $V'$, in the construction, they can only be added to $A$. For a vertex $v'_i$ in $V'$, if $v_iv_j$ is an edge in $E(G)$ where $i<j$, then $v'_i$ has neighbors $e^1_{ij}, e^4_{ij}$, and if $v'_i$ is added to $A$, then $e^1_{ij}, e^4_{ij}$ are also added to $A$ (step 2). Further, if $v_jv_i$ is an edge in $E(G)$ where $j<i$, then $v'_i$ has neighbors $e^2_{ji}, e^3_{ji}$, and if $v'_i$ is added to $A$, then its neighbor $v'_j$ is not added to $A$ (step 1), and in step 2 we added $e^2_{ji}, e^3_{ji}$ to $A$.
     Further, if a vertex $v'_i$ is added to $A$, then we also added its neighbor $v''_i$ to $B$, and every vertex $v'_i$ in $V'$ has only one neighbor $v''_i$ in $V''$. Thus, if a vertex of $V'$ is added to $A$, then it has exactly one neighbor in $B$.
     \item Consider the vertices of $V''$, in the construction, they can only be added to $B$. We added $v''_i$ to $B$ whenever we added its neighbor $v'_i$ to $A$. All the vertices of $V''$ has degree $1$. Thus, if a vertex of $V''$ is added to $B$, it has exactly one neighbor in $A$.
     \item Consider the vertices of $X_e$, no vertex in $X_e$ is connected to a vertex in $V''$. For every $v_iv_j$  in $E(G)$ where $i<j$, vertices $e^1_{ij}$ and $e^4_{ij}$ are connected to only $v'_i$ in $V'$ and $e^2_{ij}$ and $e^3_{ij}$ are connected to only $v'_j$ in set $V'$. If $v'_i$ was added to $A$ ($v'_j$ was not added to $A$),  then we added both $e^1_{ij}$ and $e^4_{ij}$ to $A$ (step 2), and $e^2_{ij}$ and $e^3_{ij}$ added to $B$. Else, if $v'_i$ was not added to $A$ ($v'_j$ may be added to $A$), then we added $e^1_{ij}$ and $e^4_{ij}$ to $B$, and added $e^2_{ij}$ and $e^3_{ij}$ to $A$. This way no vertex of $X_e$ has a neighbor from $V'$ in its opposite set in $(A,B)$. Further, for every edge $v_iv_j\in E(G)$, when $e^1_{ij}$,$e^4_{ij}$ are added to $A$ (resp. $B$), then $e^2_{ij}$,$ e^3_{ij}$ are added to $B$ (resp. $A$). Thus, $e^1_{ij}$ and $e^2_{ij}$ are the only neighbors of each other in opposite sets, similarly $e^3_{ij}$ and $e^4_{ij}$ are the only neighbors of each other in opposite sets. Thus, every vertex in $X_e$ has exactly one neighbor in its opposite set in $(A,B)$.
 \end{itemize}
 
 The above conclude that every vertex in $A$ (resp. $B$) has exactly one neighbor in $B$ (resp. $A$), and $(A,B)$ is a pair of perfectly matched sets. To bound $|E(A,B)|$, it will suffice to bound $|A|$ as $(A,B)$ is a pair of perfectly matched sets. Observe that we are adding $|I|=k$ vertices from $V'$ in $A$ (step 1), and for every edge $v_iv_j$ in $E(G)$, we are adding two vertices from $X_e$ in $A$ (step 2). Thus, $|A|= k+2\cdot |E(G)|$.

For the other direction, let $(A,B)$ be a pair of perfectly matched sets in $G'$ such that $|E(A,B)|= k+2\cdot |E(G)|$. For the proof of this direction, we will modify the set $A$ and/or $B$ while maintaining that $(A,B)$ remains a pair of perfectly matched sets and that $|E(A,B)|$ does not decrease. For a pair $(A,B)$ of perfectly matched sets, we say a vertex $x\in A$ (resp. $B$) is \textit{matched} to $y\in B$ (resp. $A$) if $x$ and $y$ are neighbors. Consider the following modifications which we will apply in the same order in which they are described, at each step, a modification is applied exhaustively.

\begin{itemize}
    \item \textbf{M1:} If there exists a vertex $v'_i\in V'\cap A$ which is matched to a vertex $v'_j\in V'\cap B$, then remove $v'_j$ from $B$ and add $v''_i$ to $B$. Observe that it is safe to do so, since $v''_i$ is connected to only $v'_i$ in $G'$ and $|E(A,B)|$ remains unchanged.
    \item \textbf{M2:} If there exists a vertex $v'_i\in V'\cap A$ (resp. $V'\cap B$) which is matched to a vertex $e^p_{ij}\in X_e \cap B$ (resp. $X_e \cap A$) where $p\in [4]$, then remove $e^p_{ij}$ from $B$ (resp. $A$) and add $v''_i$ to $B$ (resp. $A$). Observe that it is safe to do so, since $v''_i$ is connected to only $v'_i$ in $G'$ and $|E(A,B)|$ remains unchanged.
\end{itemize}

We apply the above modifications exhaustively and due to which, for every vertex $v'_i\in V'$, if $v'_i$ is in $(A\cup B)$, then $v'_i$ is matched to $v''_i$, and thus no two neighbors in $V'$ can belong to opposite sets in $(A,B)$. Further, if two distinct vertices $v'_i,v'_j\in V'\cap B$ (resp. $V'\cap A$) are neighbors in $G'$, then none of the vertices from $\{e^1_{ij},e^2_{ij},e^3_{ij},e^4_{ij}\}$ can belong to $A$ (resp. $B$) without violating any property of perfectly matched sets, and they can not be matched by $v'_i$ or $v'_j$, and hence none of them belongs to either $A$ or $B$. Thus, we further modify $(A,B)$ as follows:

\begin{itemize}
    \item \textbf{M3:} If two distinct vertices $v'_i,v'_j\in V'\cap A$ (resp $V'\cap B)$ are neighbors, then we remove  $v'_i,v'_j$ from $A$ (resp. $B$) and remove $v''_i,v''_j$ from $B$ (resp. $A$), we then put $e^2_{ij},e^3_{ij}$ in $B$ and $e^1_{ij},e^4_{ij}$ in $A$. Observe that this modification maintains that $(A,B)$ remains perfectly matched sets and $|E(A,B)|$ remains unchanged.
\end{itemize}

Exhaustive application of the above modification ensures that no two neighbors in $V'$ belong to $A\cup B$. Recalling that $|E(A,B)|=k+2\cdot |E(G)|$, and that any vertex from $X_e$ if belongs to $A\cup B$ is matched to a vertex from $X_e$. Since there are at most $4\cdot |E(G)|$ vertices in $X_e$, they contribute at most $2\cdot |E(G)|$ to the $|E(A,B)|$. Thus, there must be at least $k$ vertices from $V'$ in $A\cup B$ as $V''$ vertices can only be matched to vertices of $V'$. As discussed earlier, no two neighbors in $V'$ belong to $A\cup B$. Thus, let $I^*= \{v_i |\ v'_i \in V'\cap (A\cup B)\}$. Recalling that $v'_i$ and $v'_j$ are adjacent in $G'$ if and only if $v_i$ and $v_j$ are adjacent in $G$, this ensures $I^*$ is an independent set in $G$ and $|I^*|=k$. This finishes the proof.
\end{proof}
\section{Conclusions}
We showed PMS to be $W[1]$-hard with respect to solution size, while obtaining FPT algorithms with respect to several structural parameters. We leave open the problem of obtaining an exponential time algorithm with time complexity significantly lower than than our bound of $O^*(1.966^n)$.
\newline
\newline
\textbf{Acknowledgements.}
We thank anonymous reviewers for their feedback and useful suggestions.
\bibliographystyle{plainurl}
\bibliography{ms}

\end{document}